\documentclass[final,12pt,notitlepage]{article}
\usepackage[margin=1in]{geometry}

\usepackage{amsfonts}
\usepackage{amsmath}
\usepackage{mathtools}
\usepackage{amsthm}
\usepackage{thmtools}
\usepackage{textcomp}
\usepackage{amssymb}
\usepackage[svgnames]{xcolor}
\usepackage[sc,labelsep=period]{caption}
\usepackage[bottom,multiple,splitrule]{footmisc}
\usepackage{graphicx}
\usepackage{csquotes}
\usepackage[authoryear,round,longnamesfirst,numbers]{natbib}
\usepackage[inline,shortlabels]{enumitem}
\usepackage[labelfont=rm]{subfig}
\usepackage{hyperref}
\usepackage[onehalfspacing]{setspace}
\usepackage{titlesec}
\usepackage{sectsty}
\usepackage{dsfont}
\usepackage[capitalize,nameinlink]{cleveref}
\usepackage[colorinlistoftodos,textsize=tiny]{todonotes}
\usepackage{xargs}
\usepackage{comment}
\usepackage{tikz}
\usetikzlibrary{positioning,chains,fit,shapes,calc,snakes, arrows.meta}
\usepackage{booktabs}
\usepackage{array} 

\usepackage{times}

\usepackage{authblk}
\usepackage{etoolbox}

\hypersetup{
pdftitle={Stability in Repeated Matching Markets},
pdfauthor={Ce Liu},
pdfkeywords={dynamic stability, repeated game, two-sided matching}
}
\hypersetup{
colorlinks,breaklinks,naturalnames,
citecolor=DarkBlue,urlcolor=DarkBlue,linkcolor=DarkBlue
}

\bibpunct[, ]{(}{)}{;}{a}{}{,}
\MakeOuterQuote{"}

\theoremstyle{plain}

\theoremstyle{definition}
\newtheorem{theorem}{Theorem}
\newtheorem{assumption}{Assumption}

\newtheorem{lemma}{Lemma}
\newtheorem{observation}{Observation}
\newtheorem*{observation*}{Observation}

\newtheorem{example}{Example}
\newtheorem{definition}{Definition}

\theoremstyle{remark}

\renewcommand\thmcontinues[1]{Continued}

\AtEndEnvironment{proof}{\setcounter{claim}{0}}

\newcommand{\hospitals}{\mathcal{F}}
\newcommand{\students}{\mathcal{W}}

\newcommand{\smatchings}{M}

\newcommand{\histories}{\mathcal{H}}

\renewcommand{\Re}{\mathbb{R}}

\newcommand{\tcs}{\mathcal{T}}
\newcommand{\reduced}{\mathcal{R}}

\newcommand{\redmatching}{\smatchings_{\reduced}}

\newcommand{\abs}[1]{ \left| #1 \right| }
\newcommand{\size}[1]{\big| #1 \big| }

\newcommand{\term}[1]{\textit{#1}}

\newcommand{\fhistories}{\overline{\histories}}

\renewcommand{\emph}[1]{\textit{#1}}
\renewcommand{\exp}{\mathbb{E}}
\renewcommand{\tilde}{\widetilde}
\renewcommand{\hat}{\widehat}

\DeclareMathOperator*{\argmin}{arg\,min}

\newenvironment{proofof}[1]{\par\bigskip
   \noindent \textbf{Proof of #1.} } {\qed \par \bigskip}

\sectionfont{\color{DarkRed}
}

\subsectionfont{\color{DarkRed}
}

\subsubsectionfont{\color{DarkRed}
}

\renewcommand{\paragraph}[1]{{\flushleft \textbf{\color{DarkRed}#1 }}}

\newcommand{\hchi}{\raise0.4ex\hbox{$\chi$}}

\setlist{topsep=3pt, itemsep=-3pt}

\makeatletter
\g@addto@macro \normalsize {%
 \setlength\abovedisplayskip{8pt plus 2pt minus 2pt}%
 \setlength\belowdisplayskip{8pt plus 2pt minus 2pt}%
}
\makeatother

\begin{document}

\title{\text{Stability in Repeated Matching Markets}}

\author{Ce Liu\thanks{Department of Economics, Michigan State University (e-mail: \href{mailto: celiu@msu.edu}{celiu@msu.edu}).} \thanks{I am grateful to Nageeb Ali for his continuing guidance and encouragement, and Joel Sobel for his advice. I thank Mohammad Akbarpour, Christopher Chambers, Henrique de Oliveira, Laura Doval, Federico Echenique, Jon Eguia, Simone Galperti, Ed Green, Fuhito Kojima, Maciej Kotowski, Shengwu Li, Elliot Lipnowski, Zizhen Ma, Arijit Mukherjee, Esteban Peralta, Mike Powell, Ran Shorrer, Thomas Wiseman, Alexander Wolitzky, Kyle Woodward, Qinggong Wu, Hanzhe Zhang, three anonymous referees at EC'19, as well as seminar and workshop participants at 2018 NSF/NBER/CEME Mathematical Economics Conference, 5th Workshop on Relational Contracts, SEA 2019, ES World Congress 2020, UCSD, NYU, HKU, HKUST, Michigan State, Northwestern Kellogg, University of Michigan, Western Michigan University, and Stanford University for helpful comments.}}

\date{\today}

\maketitle
\begin{abstract}
This paper develops a framework for repeated matching markets. The model departs from the Gale-Shapley matching model by having a fixed set of long-lived players (hospitals) match with a new generation of short-lived players (residents) in every period. I show that there are two kinds of hospitals in this repeated environment: some hospitals can be motivated dynamically to voluntarily reduce their hiring capacity, potentially making more residents available to rural hospitals without violating stability; the others, however, are untouchable even with repeated interaction and must obtain the same match as they do in a static matching. In large matching markets with correlated preferences, at most a vanishingly small fraction of the hospitals are untouchable. The vast majority of hospitals can be motivated using dynamic incentives.


%
\end{abstract}


\setcounter{page}{1}
\setstretch{1.25}

\clearpage


\section{Introduction}
A 2017 study by the US Government Accountability Office\footnote{https://www.gao.gov/products/GAO-17-411.} projects a deficit of over 20,000 primary care physicians in rural hospitals by 2025, and points to the disproportionate concentration of medical residents in urban hospitals as a key contributing factor to this problem. Historically, this workforce shortage in rural hospitals has lead to lower medical utilization and worse health conditions of the rural population.\footnote{See, for example, \cite{bindman1990}, \cite{hadley1991} and \cite{rosenbach1995}, and a recent coverage by NPR: https://www.npr.org/sections/health-shots/2019/05/21/725118232/the-struggle-to-hire-and-keep-doctors-in-rural-areas-means-patients-go-without-c.} Despite significant policy interests in increasing the number of medical residents at rural hospitals, it has proved difficult to provide the extra funding needed for rural hospitals to be competitive against their urban counterparts.


One promising solution is to re-design the National Resident Matching Program (NRMP), so that it matches more residents to rural hospitals (see, for example, Chapter 5.4.4 of \citealp{rothsotomayor1992}). However, according to the celebrated Rural Hospital Theorem \citep{roth1986rural}, doing so will always end up creating instabilities, and jeopardizing the operation of the hospital-resident matching market. The Rural Hospital Theorem states:
\begin{quote}
\textit{If preferences over individuals are strict, any hospital that does not fill its quota at some (static) stable matching is assigned precisely the same set of students at every (static) stable matching.} 
\end{quote}
If a rural hospital fails to fill its hiring slots at one static stable matching, it is then impossible for the NRMP to create a different static stable matching that gives more or better matches to this hospital.

However, this interpretation of the Rural Hospital Theorem is based on an assumption that treats the market as a one-shot interaction. In reality, hospital-resident matching is an ongoing process that takes place every year: hospitals are long-lived players that hire from this market in every recruitment cycle, whereas students are short-lived players that participate in this process only once. While only static matchings can be created in one-shot environments, in repeated matching markets, a matching clearing house can create different future matchings depending on what has  happened in the past. This way, it can effectively create history-dependent matching processes that are not subject to the implication of the Rural Hospital Theorem. 

An example of such dynamic processes already exists in the Japanese Resident Matching Program: if a hospital hires against the recommendation from the matching clearing house, it is excluded from future participation in the market.\footnote{See \cite{kojimakamada2015} for more details on the Japanese matching market.} Thus a different future matching is created to punish the hospital for its past deviation.  In the US, the NRMP is a non-profit organization that lacks the legal authority to externally enforce compliance from the hospitals, so any matching process it uses must be self enforcing. The goal of this paper is to investigate what can be achieved through history-dependent matching processes that are self enforcing. In particular, I will focus on whether it is possible to reduce urban hospitals' hiring quotas without violating stability, thereby making more residents available to rural hospitals.

To this end, I consider a two-sided one-to-many matching model with long-lived hospitals and short-lived students. In each period, a new generation of students enter the market and live for one period. The stage game is the canonical one-to-many matching game \`{a} la \cite{galeshapley1962} among hospitals and students who are currently active. I define a stability notion---self-enforcing matching process---that generalizes static stability to this repeated environment. Specifically, a matching process is a complete contingent plan that specifies a current stage-game matching as a function of past histories. A blocking coalition in the stage game can comprise a hospital and a set of students who also find this deviation profitable, but unlike the theory of static matching, hospitals take the continuation play into account, and are unwilling to deviate if doing so leads to unattractive future outcomes. A matching process is said to be self enforcing if it is immune to not only unilateral deviations by hospitals or students, but also sequential blocking coalitions that can be chained together by a hospital over a possibly infinite horizon.

The results in this paper identify two kinds of hospitals in repeated matching markets: some hospitals can be motivated to voluntarily reduce their hiring quotas through continuation play, making more residents available for rural hospitals; but there are also other hospitals that are untouchable by dynamic enforcement. I then move on to large repeated matching markets to characterize the relative sizes of these two kinds of hospitals, and show that the vast majority of hospitals can indeed be motivated dynamically.



%

\paragraph{Repeated Matching.}
%
\cref{section: repeated market} studies a setting where the same stage-game matching is played repeatedly in every period, and contrasts it with the benchmark static matching environment where the stage-game is treated as a one-shot interaction. To highlight the effects of repeated interactions, in this section I will assume that the same set of representative students are born every period; random populations will be considered later on in \cref{section: large markets}.

In \cref{subsection: top coalition sequence}, I identify a set of players named the \textit{top coalition sequence}: a hospital and a group of students form a \term{top coalition} in the stage game if they are mutual favorites; once identified, a top coalition is removed from the stage game, and new top coalitions are searched for among the remaining players. The top coalition sequence is constructed by carrying through this iterative process until there is no new top coalition. I find that the top coalition sequence defines a partition over the set of players. 

\cref{theorem:tcs} shows that hospitals and students in the top coalition sequence are untouchable even with repeated play and history dependence. A top coalition hospital is always matched to its top coalition students in every static stable matching. In repeated matching markets, they remain so
in every self-enforcing matching process and at every possible history of the market; this statement also holds regardless of the patience level of the hospitals. 

\cref{theorem:tcs} has a simple intuition. A hospital in a top coalition is immune from any future repercussions for its current behavior: in every future generation, it will be the favorite hospital of its favorite students, so there is no credible way to separate this hospital from its top coalition students in the future. As a result, it is impossible to motivate this hospital to give up its top coalition students in the current period through continuation value. Since this arguments holds at every history of the market and regardless of patience, we can effectively treat players in the top coalition as inactive players. By iteratively applying the same argument to the remaining active players, an inductive procedure extends this logic to the entire top coalition sequence.

In contrast to \cref{theorem:tcs}, \cref{theorem:folk-theorem} proves a folk theorem for the hospitals outside of the top coalition sequence. When defining the minmax payoff for hospitals outside of the top coalition sequence, it is crucial to make sure that the top coalition sequence is left intact, and calibrate the minmax payoffs accordingly. \cref{theorem:folk-theorem} then shows that if a (randomized) stage-game matching leaves the top coalition sequence intact, respects students' participation constraints, and guarantees all other hospitals higher than their appropriately defined minmax payoffs, then it can be sustained as the stationary outcome of a self-enforcing matching process when hospitals are patient.

As we will see in the example in \cref{section: example}, in certain matching markets the top coalition sequence may be empty, but under some other preference configurations the top coalition sequence may span the entire set of players. In the matching market for medical residents, hospitals' valuations of students are positively correlated, while students also have positively correlated preferences over hospitals.  It is then crucial to ask: in these markets, how big is the impact from the untouchable hospitals?

\paragraph{Large Repeated Matching.} To answer this question, in \cref{section: large markets}, I build on the repeated matching model in \cref{section: repeated market}, but allow students to be drawn randomly from a fixed distribution in every period. Preferences are correlated: both hospitals and students can be divided into finite quality classes; players' payoffs are determined by their match partners' quality class along with an idiosyncratic shock. I will let market size grow to infinity and focus on large market asymptotics.

In this setting, I will call the best class of hospitals \textit{elite hospitals} if the size of this quality class is vanishingly small relative to the size of the best class of students.  \cref{theorem: cannot reduced capacity} shows that for every fixed discount factor, as market size grows large, all elite {hospitals} (should they exist) must always hire at full capacity in every self-enforcing matching process. In other words, {elite hospitals} are untouchable in large matching markets. 

To understand the intuition for \cref{theorem: cannot reduced capacity}, note that an elite hospital faces not only no vertical competition, but also vanishing horizontal competition from other elite hospitals. Such ``market power'' means that despite the randomness in preference realizations, an elite hospital can secure a high continuation value no matter what. At any fixed patience level, as market size grows large, this payoff guarantee eventually becomes high enough so that even the worst future punishment cannot justify giving up a hiring quota in the current period.


 
\cref{theorem: cannot reduced capacity} is reminiscent of \cref{theorem:tcs}, and highlights the difficulty in disciplining certain hospitals in the markets. However, elite hospitals, by definition, must make up only a vanishingly small fraction of the market. By contrast, \cref{theorem: can reduce capacity} confirms that as long as a hospital quality class makes up a non-vanishing fraction of the market, we can construct a matching process that simultaneously reduces the hiring capacity of all hospitals in this quality class; this matching process is also self-enforcing when market size is sufficiently large and hospitals are sufficiently patient. Perhaps surprisingly, \cref{theorem: can reduce capacity} also applies to the top quality class as long as its size is non-vanishing: in this case, there are no elite hospitals in the market. 

A direct implication of \cref{theorem: can reduce capacity} is that almost all hospitals can be motivated to voluntarily reduce their hiring capacities through history dependence, regardless of how they are vertically differentiated. To the extent that policy interests lie in allocating more students to rural hospitals, the untouchable hospitals would only make a negligible impact on the number of students that can be made available to the rural hospitals.


One obstacle to proving the results above is that players may derive approximately efficient payoffs from all static stable matchings (see \citealp{pittel1989, pittel1992} and  \citealp{lee2016} for large market results on this point; \citealp{ashlagi2017} show that this is true even in small matching markets). This makes it difficult to use static stable matchings as punishments in the proof of  \cref{theorem: can reduce capacity}. Instead, I construct a variant of the student-proposing serial dictatorship algorithm, and show that it can punish hospitals effectively even in large matching markets, while also being  robust to different preference correlation structures. The algorithm leverages horizontal competition among hospitals in the same quality class to achieve the desired punishments, regardless of how these hospitals are vertically differentiated.

\subsection{Related Literature} \label{section: literature}

This paper is related to several different lines of research. First, my paper is part of a large and active literature on community enforcement, which studies how repeated interactions can lead to desirable outcomes that are not sustainable in one-shot interactions. See, for example, \cite{kandori1992social}, \cite{ellison1994cooperation}, \cite{wolitzky2013cooperation}, \cite{ali2016ostracism}, \cite{acemogluwolitzky1,acemogluwolitzky2}, and \cite{debwolitzky}. The main difference of my paper is that I focus on two-sided one-to-many matching environments, which may contain ``top coalition'' players that cannot be motivated dynamically even when patience is high. In order to quantify the impact of these untouchable players, I build on techniques from the large matching market literature to obtain asymptotic characterizations of their relative size in the market.



Second, this paper is also related to the literature on dynamic matching. 
\citet{du2016rigidity} and \citet{doval2018} consider the existence of self-enforcing arrangements in a setting where matching is one-to-one, and players leave the market permanently once matched.\footnote{See \cite{unver2010dynamic}, \cite{anderson2015}, \citet{baccara2015}, \citet{leshno2017dynamic}, and \cite{akbarpour2017} for the welfare implications of various dynamic matching algorithms in such markets.}  
  Another strand of this literature investigates self enforcing arrangements in matching markets where the links among long-lived players can be revised over time. See, for example, \cite{corbae2003}, \cite{damianolam2005}, \cite{kurino2009}, \cite{newtonsawa2015},  \cite{kadamkotowski20182, kadamkotwoski2018}, \cite{kotowski2015note}, and \cite{altinok2019dynamic}. The main difference of the current paper is that I study a setting where a fixed set of long-lived players match with multiple short-lived players in every period. As a result of this difference, while dynamic incentives typically impede either stability or efficiency in the existing literature, in my paper they are used as carrot and stick to enforce more stable outcomes.


Third, the current paper is also part of a nascent literature combining repeated games and cooperative games. \cite{bernheimslavov2009} study a repeated version of Condorcet Winner. \cite{aliliu} consider coalitional deviations in general repeated games where the stage game can either be a strategic-form game or a cooperative game. The solution concept in my paper builds on the full history dependence and subgame-perfection requirement in these papers. The main difference is in the form of effective coalitions: in both these papers, the effective coalitions consist of subsets of long-lived players; in the current paper, however, the effective coalitions are those that consist of a single long-lived player and multiple generations of short-lived players. \cite{aliliu} also focus on the effects of public versus secret payments, whereas the current paper focuses on matching markets without transfers. More recently, \cite{bardhi2020early} build on the solution concept in \cite{aliliu} to study early career discrimination in matching markets where wages are flexible. 

%

\section{Example} \label{section: example}

In this section, I illustrate the effect of history dependence through a stylized example based on the matching market between hospitals and medical residents.

Suppose there are three hospitals: two of these hospitals, $f_1$ and $f_2$, are urban hospitals, while $f_r$ is rural ({I reserve the letter ``$h$'' for a history, which will be introduced later. Instead, I use the letter $f$ to represent a hospital as a firm, and $w$ for a student who is also a worker}). Hospitals are long-lived players, each with two hiring slots to fill every year. On the other side of the market, five representative students $w_1, \ldots, w_5$ enter the market looking for residency jobs each year.  Students are short-lived players.

The left panel of \cref{table:example preferences} is hospitals' stage-game utilities. For this example, we shall assume that hospitals have additively separable utilities from matched students, and derive $0$ from unfilled positions. Observe that $w_5$ is every hospital's least preferred student. Students' preference over hospitals are in the right panel of \cref{table:example preferences}. Each student prefers working for any hospital over unemployment. Observe that the rural hospital $f_r$ is the worst hospital for every student. 

\[
\begin{array}{c | c c c c c} 
\toprule
u_f(w) & 5  &  4  & 3 & 2 & 1\\
\midrule
f_1 & w_1 & w_2 & w_3 & w_4 & w_5\\
f_2 & w_3 & w_2 & w_4 & w_1 & w_5\\
f_r & w_2 & w_4 & w_3 & w_1 & w_5\\
\bottomrule
\end{array}
\hspace{10ex}
\begin{array}{c | c c c} 
\toprule
 &   & \succ &  \\
\midrule
w_1 & f_2 & f_1 & f_r \\
w_2 & f_2 & f_1 & f_r \\
w_3 & f_1 & f_2 & f_r \\
w_4 & f_1 & f_2 & f_r \\
w_5 & f_1 & f_2 & f_r \\
\bottomrule
\end{array}
\]
\begingroup
\captionof{table}{Example: Preferences\label{table:example preferences}}
\endgroup

\paragraph{Static Stable Matchings.} If the market in each period is treated as an isolated one-shot interaction, then there are two static stable matchings, shown in \Cref{figure:example matchings}. Both of them match the rural hospital $f_r$ with the worst student $w_5$ while leaving its other position unfilled.

\begin{figure}[h]  \centering
\minipage{0.32\textwidth} \centering
\begin{tikzpicture}
\tikzstyle{every node} = [draw, shape=circle, line width=1pt, inner sep=0pt, minimum size=20pt]
\tikzset{every path/.append style={line width=1pt}}
\path (0,3) node (f1) {$f_1$};
\path (0,2) node (f2) {$f_2$};
\path (0,1) node (fr) {$f_r$};
\path (3,4) node (w1) {$w_1$};
\path (3,3) node (w2) {$w_2$};
\path (3,2) node (w3) {$w_3$};
\path (3,1) node (w4) {$w_4$};
\path (3,0) node (w5) {$w_5$};
\draw (f1) -- (w1);
\draw (f1) -- (w3);
\draw (f2) -- (w2);
\draw (f2) -- (w4);
\draw (fr) -- (w5);
\end{tikzpicture}	
\caption*{$m_{\hospitals}$}
\endminipage 
\hspace{10ex}
\minipage{0.32\textwidth} \centering
\begin{tikzpicture}
\tikzstyle{every node} = [draw, shape=circle, line width=1pt, inner sep=0pt, minimum size=20pt]
\tikzset{every path/.append style={line width=1pt}}
\path (0,3) node (f1) {$f_1$};
\path (0,2) node (f2) {$f_2$};
\path (0,1) node (fr) {$f_r$};
\path (3,4) node (w1) {$w_1$};
\path (3,3) node (w2) {$w_2$};
\path (3,2) node (w3) {$w_3$};
\path (3,1) node (w4) {$w_4$};
\path (3,0) node (w5) {$w_5$};
\draw (f1) -- (w3);
\draw (f1) -- (w4);
\draw (f2) -- (w1);
\draw (f2) -- (w2);
\draw (fr) -- (w5);
\end{tikzpicture}
\caption*{$m_{\students}$}
\endminipage 
\caption{Static Stable Matchings\label{figure:example matchings}}
\end{figure}

Suppose for the welfare of the patients living in rural areas, we want to match $f_r$ with $w_2$ and create the matching $m_0$ in \Cref{figure:unstable matching}. But $m_0$ is unstable: $f_1$ and $w_2$ will form a blocking pair. 

\begin{figure}[h] \centering
\minipage{0.32\textwidth} \centering
\begin{tikzpicture}
\tikzstyle{every node} = [draw, shape=circle, line width=1pt, inner sep=0pt, minimum size=20pt]
\tikzset{every path/.append style={line width=1pt}}
\path (0,3) node (f1) {$f_1$};
\path (0,2) node (f2) {$f_2$};
\path (0,1) node (fr) {$f_r$};
\path (3,4) node (w1) {$w_1$};
\path (3,3) node (w2) {$w_2$};
\path (3,2) node (w3) {$w_3$};
\path (3,1) node (w4) {$w_4$};
\path (3,0) node (w5) {$w_5$};
\draw (f1) -- (w1);
\draw (f1) -- (w5);
\draw (f2) -- (w3);
\draw (f2) -- (w4);
\draw (fr) -- (w2);
\end{tikzpicture}	
\endminipage 
\caption{$m_0$: An Unstable Matching\label{figure:unstable matching}}
\end{figure}
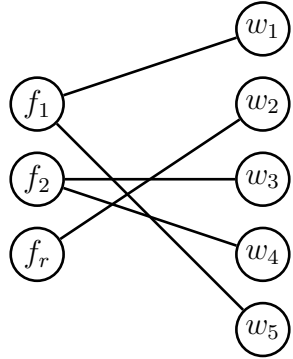

\paragraph{A Self-Enforcing Matching Process.}
Now suppose the matching clearing house can vary the matching it creates in the future based on how hospitals matched in the past. Consider the matching process $\mu^0$ in \Cref{figure:example triggering}: hospitals and students start from the state on the left, matching according to $m_0$ in every period.  If any hospital deviates, players will move to the state on the right, which is an absorbing state. In this state, players match according to the student optimal matching $m_{\students}$ in every period. Note that this is the familiar idea of Nash Reversion, but in this case, the stage game is a cooperative game. Hospitals' stage-game payoffs from $m_{\students}$ and $m_0$ are included in \Cref{figure:example triggering}.

If the market evolves according to $\mu^0$, then the stage-game matching $m_0$ is played in every period. I will show that $\mu^0$ is self-enforcing when patience is high. A one-shot deviation principle, established later in \cref{lemma: one-shot-deviation}, implies that a matching process is self-enforcing if and only if two requirements are satisfied at every history of the market: 
\begin{enumerate}
	\item no student wishes to unilaterally leave her matched hospital, and
	\item no hospital finds it profitable to conduct a one-shot deviation  with a group of students who also find this deviation profitable.
\end{enumerate}

\begin{figure}[h]  \centering
\makebox{
\includegraphics[width=4.5in]{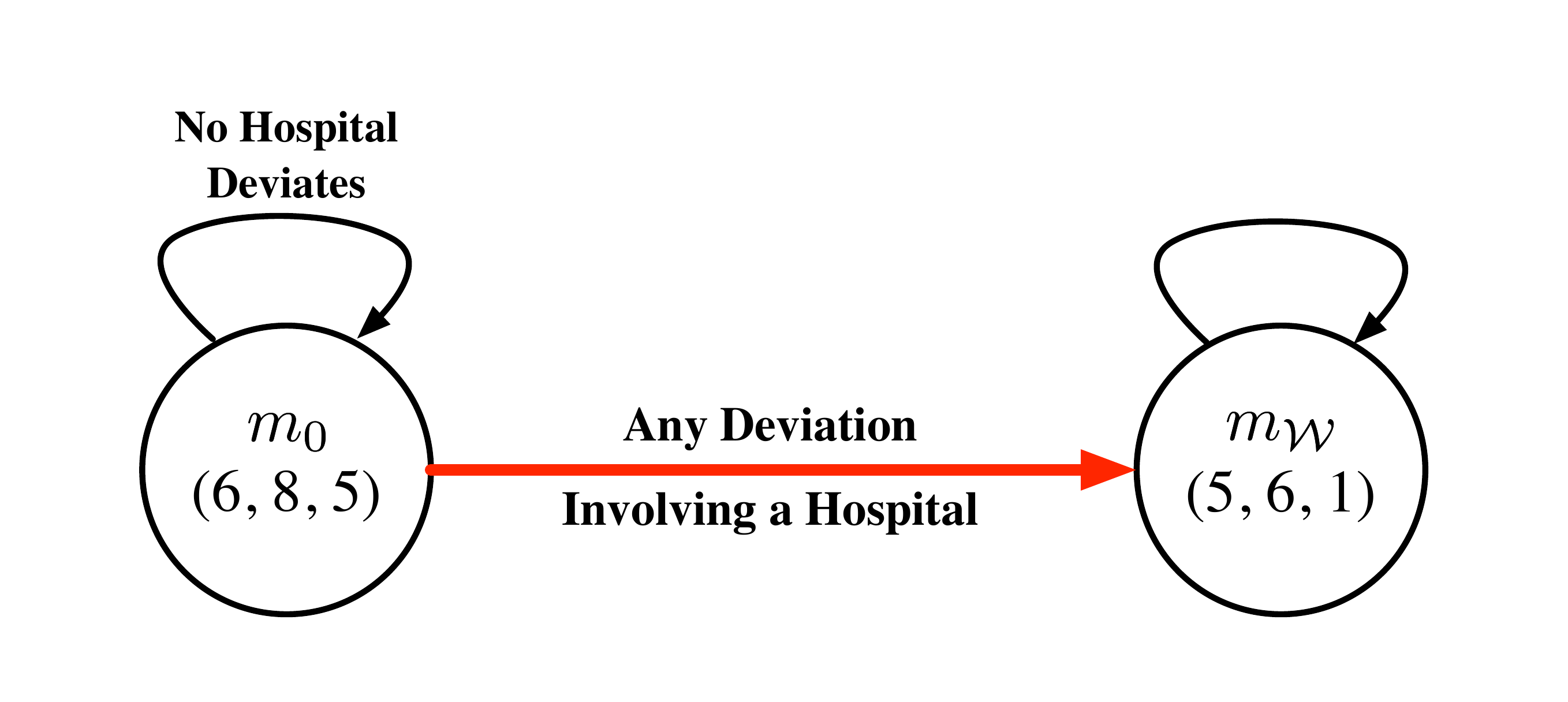}
}
\caption{$\mu_0$: A Self-Enforcing Matching Process Implementing $m_0$ \label{figure:example triggering}}
\end{figure}

%

The state on the right is self-enforcing since it is the infinite repetition of a static stable matching: by construction, there are no profitable deviations. In the state on the left, since every student prefers employment, no students would wish to unilaterally leave their hospital.  By following $\mu^0$, $f_1$ receives $6$ every period from matching with $\{ w_1, w_5\}$. If $f_1$ were to deviate with any student, it receives no more than $9$ in the current period, and $5$ from all future periods. As long as $f_1$'s discount factor $\delta$ satisfies $(1-\delta )9 + \delta 5 < 6$, or $\delta >{3}/{4}$, $f_1$ would not find such one-shot deviations profitable. A similar argument rules out any profitable one-shot deviations involving $f_2$. In light of \cref{lemma: one-shot-deviation}, $\mu^0$ is a self-enforcing matching process for $\delta >{3}/{4}$.

\paragraph{Limits of Self Enforcement.}
The analysis so far has shown it's possible to use history dependence to expand the set of stable outcome. The next example will illustrate that certain preference configurations can severely limit this ability.
\[
\begin{array}{c | c c c c c} 
\toprule
u_f(w) & 5  &  4  & 3 & 2 & 1\\
\midrule
f_1 & w_1 & w_2 & w_3 & w_4 & w_5\\
f_2 & w_3 & w_2 & w_4 & w_1 & w_5\\
f_r & w_2 & w_4 & w_3 & w_1 & w_5\\
\bottomrule
\end{array}
\hspace{10ex}
\begin{array}{c | c c c} 
\toprule
 &   & \succ &  \\
\midrule
w_1 & f_1 & f_2 & f_r \\
w_2 & f_1 & f_2 & f_r \\
w_3 & f_1 & f_2 & f_r \\
w_4 & f_1 & f_2 & f_r \\
w_5 & f_1 & f_2 & f_r \\
\bottomrule
\end{array}
\]
\begingroup
\captionof{table}{Students Share Identical Preference \label{table:matchings aligned}}
\endgroup\vspace{0.25in}
\vspace{-0.05in}

Consider the market in \cref{table:matchings aligned}: the only difference from \cref{table:example preferences} is that now all students share a common preference ranking $f_1 \succ f_2 \succ f_r$ over hospitals. There is a unique static stable matching $m^*$ in this market, as depicted in \Cref{figure:unique stable matching}. In this market, however, no self-enforcing matching process can sustain any matching other than $m^*$ regardless of how patient the hospitals are.

\begin{figure}[h] \centering
\minipage{0.32\textwidth} \centering
\begin{tikzpicture}
\tikzstyle{every node} = [draw, shape=circle, line width=1pt, inner sep=0pt, minimum size=20pt]
\tikzset{every path/.append style={line width=1pt}}
\path (0,3) node (f1) {$f_1$};
\path (0,2) node (f2) {$f_2$};
\path (0,1) node (fr) {$f_r$};
\path (3,4) node (w1) {$w_1$};
\path (3,3) node (w2) {$w_2$};
\path (3,2) node (w3) {$w_3$};
\path (3,1) node (w4) {$w_4$};
\path (3,0) node (w5) {$w_5$};
\draw (f1) -- (w1);
\draw (f1) -- (w2);
\draw (f2) -- (w3);
\draw (f2) -- (w4);
\draw (fr) -- (w5);
\end{tikzpicture}	
\endminipage 
\caption{$m^*$: The Unique Stable Matching\label{figure:unique stable matching} }
\end{figure}
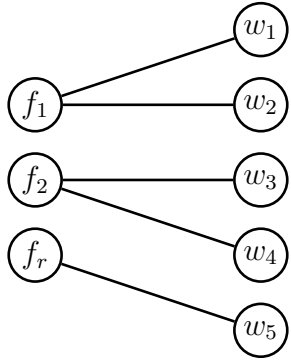

 To see why, first observe that as everyone's favorite hospital, $f_1$ finds its favorite students $\{w_1, w_2\}$ available in every future generation: whenever $f_1$ is not matched to $\{w_1, w_2\}$, it can always poach them. Since $\{w_1, w_2\}$ also happen to give $f_1$ the highest possible stage-game payoff, it is therefore impossible to either punish or reward $f_1$ through continuation value. Without changes in continuation value, $f_1$ acts like a short-lived player, and will always match with $\{w_1, w_2\}$ in every self-enforcing matching process at every history of the market.

%

Now since $\{f_1, w_1, w_2\}$ are always matched together, they can be disregarded by other players. This makes $f_2$ every student's favorite hospital among the remaining hospitals, and $f_2$ finds its favorite students among the remaining students, $\{w_3, w_4\}$, always available in every future generation. The only way to credibly remove $w_3$ or $w_4$ from $f_2$ is to match them with $f_1$---otherwise $f_2$ can simply poach them back. But this isn't possible since $f_1$ is always occupied by $\{w_1, w_2\}$ at every history. Without changes in continuation value, $f_2$ also behaves myopically, and matches with $\{w_3, w_4\}$ at every history. A similar ``peeling'' argument along students' shared preference list ensures that $f_r$ is always matched with $w_5$ in every self-enforcing matching process. 

None of the arguments above involved the hospitals' patience, so for all $0<\delta<1$, the only self-enforcing matching process is the one where $m^*$ is played after every history. Even with history dependence, the market functions like a one-shot interaction.

\paragraph{Takeaway.} From the examples above, we see that depending on players' preference configuration, history dependence may or may not be able to expand the set of stable outcomes. It's worth noting that the uniqueness of static stable outcome is {not} responsible for the collapse of dynamic enforcement: in fact, if the hospitals share a common ranking over students, the market will have a unique static stable matching. But in this case it is possible to expand the set of stable outcomes through history dependence (see \cref{observation:common-ranking}). 

The results in \cref{section: repeated market} show that in general, there is a partition over the set of players: in one element of the partition, players may obtain matches that are unavailable in static stable matchings if hospitals are patient; in the other element of the partition, however, players always match in the same way as they do in a static stable matching regardless of hospitals' patience.  The results in \cref{section: large markets} characterize asymptotically the size and composition of these two kinds of players when market size grows to infinity.



\section{Repeated Matching Market} \label{section: repeated market}
In this section, I first review the benchmark static matching environment, then extend the model to repeated matching markets. I also introduce the notion of top coalition sequence, and show that it determines whether or not a hospital can be motivated through continuation play.

\subsection{Model}
\paragraph{{Players}.}
At the beginning of every period $t = 0, 1, 2\ldots$, a new generation of {students} $\students$ 
enter the market to match with a fixed set of {hospitals} $\hospitals$.
Hospitals are long-lived players that persist through time. Students are short-lived and remain in the market for only one period. Matching is one-to-many: each hospital $f$ has $q_f > 0$ hiring slots to fill in every period.

Each hospital $f$ has utility function $\tilde{u}_f: 2^{\students} \rightarrow \Re$ over subsets of students. Hospitals' utility functions are strict \big($\tilde{u}_f(W) = \tilde{u}_f(W') $ only if $W = W'$\big) and responsive: for all $f\in \hospitals$, $W \subseteq \students $, $w \in W$ and $w' \notin W$,
\[
	\tilde{u}_f\big (W\cup \{w' \}\backslash\{w\} \big ) > \tilde{u}_f \big (W\big)  \;\;\;\;\text{ if } \;\;\;\; \tilde{u}_f \big( \{w' \} \big) > \tilde{u}_f \big( \{w \}\big).
\]

Hospitals share a common discount factor $\delta$, and evaluate a sequence of flow utilities through exponential discounting. Each student $w$ has a strict preference relation $\succ_w$ over the set of hospitals and being unmatched (being unmatched is denoted $w$). I write $f \succeq_w f'$ if either $f \succ_w f'$ or $f = f'$.

\paragraph{{Stage Game}.}
The stage game in every period is a static one-to-many matching game played between the hospitals and the students who are active in that period. Formally, a stage-game matching $m$ is a mapping defined on the set $\hospitals \cup \students$ such that: (i) for every $w\in \students$, $m(w) \in \hospitals \cup \{ w\}$; (ii) for every $f\in \hospitals$, $m(f) \subseteq {\students}$ and $\abs{m(f)} \le q_f $; and (iii) $w \in m(f)$ if and only if $m(w) = f$. Let $\smatchings$ denote the set of all stage-game matchings. For each $f\in \hospitals$, let $u_f:\smatchings \rightarrow \Re\,$  be hospital $f$'s utility function over stage-game matchings induced from its preference over students: $u_f(m) \equiv \tilde{u}_f(m(f))$ for all $m\in \smatchings$.


A stage-game matching $m$ is subject to deviations from three kinds of coalitions: (i) a singleton coalition $\{f\}$ where hospital $f$ unilaterally fires a subset of its employees; (ii) a singleton coalition $\{w\}$ where student $w$ leaves her employer and remain unmatched; and (iii) a coalition $\{f\} \cup W$ with $|W | \le q_f$, where hospital $f$ and $W$ match together, abandoning any other pre-existing match partners. A coalitional deviation is profitable if all members of the coalition prefer the deviation to the original matching $m$. A stage-game matching is \term{stable} if none of these three kinds of coalitional deviations can be profitable; it is \term{individually rational} if no singleton coalitions are profitable.\footnote{This notion of stability is stronger than pairwise stability, but weaker than the core. These different solution concepts coincide in static matching environments where preferences satisfy substitutability (see, for example, \citealt{rothsotomayor1992}).} 


In static matching models, there is no need to specify the resulting matching outcome after a deviation. This is because in static matching environments, the profitability of a deviation does not depend on how others respond. However, in repeated matching markets where the past influences the future, we have to specify what outcome is realized after a deviation. 

To this end, let $[m,(f,W)] \in M$ denote the resulting stage-game matching after coalition $\{f\}\cup W$ deviates from stage-game matching $m$. I make the following assumption:
\begin{assumption} \label{assumption: no further deviation}
The stage-game matching $m' = [m, (f, W) ]$ satisfies $m'(f) = W$, and ${m}' (f') = m (f')\backslash W $ for all $f'\ne f$.
\end{assumption}
\cref{assumption: no further deviation} states that members of the deviating coalition are matched together in the resulting stage-game matching; in addition, players abandoned by the deviators remain unmatched, while those untouched by the deviation remain matched as before. One can make alternative assumptions that specify how other players may further deviate within the period after the initial deviation. As long as it is possible to identify the hospital that {initiated} the first deviation, identical results follow.\footnote{\cref{assumption: no further deviation} is only needed for establishing \cref{lemma: identifiability of manipulator} in \cref{section: supp prelim}: it is always possible to identify the hospital responsible for a deviation. As the main application of the current paper, the NRMP keeps extensive records of past matching outcomes. 
}

\paragraph{Repeated Matching Game.} The timing in each period is as follows: at the beginning of the period, a realization $\omega\in \Omega$ is drawn from a public randomization device (by, for example, a centralized matching clearing house); based on $\omega$ and the history of past interactions, a recommended stage-game matching is created for the active players in the market; hospitals and students then decide whether to deviate from this recommendation, which leads to the realized stage-game matching. Note that the public randomization device is \textit{not} intended to represent the random realization of players' preferences.\footnote{In fact preferences are assumed to be pre-determined so far; in \cref{section: large markets}, I augment the stage-game to account for the random preferences realizations.} Instead, it represents the ability of the matching clearing house to randomize over its recommendations.

A $t$-period ex ante histories $h = (\omega_\tau, m_\tau)_{\tau = 0}^{t-1}$ specifies a sequence of past realizations from the public randomization device and matching outcomes, before the randomization at the beginning of period $t$ is drawn. I use $\fhistories_t$ to denote the set of all $t$-period ex ante histories, with $\fhistories_0 = \{\emptyset \}$ the singleton set comprising the initial null history. Let $\fhistories \equiv \cup_{t=0}^\infty \fhistories_t$ be the set of all ex ante histories. $\histories_t \equiv \fhistories_t \times \Omega$ is the set of $t$-period ex post histories, and $\histories \equiv \fhistories \times \Omega$ the set of all ex post histories.

A \term{matching process} proposes a stage-game matching following each ex post history: a {matching process} $\mu$ is  a mapping $\mu: \histories \rightarrow M$. In the context of the NRMP, one can interpret a matching process as proposals from a history-dependent matching protocol. I use $\mu(f|h)$ and $\mu(w|h)$ to denote the match partners for hospital $f$ and student $w$ in the stage-game matching $\mu(h)$, respectively. 

Let $\fhistories_\infty = (\Omega\times M)^\infty$ be the set of {outcomes} of the repeated matching game. For an outcome $h \in \fhistories_\infty$, let $m_{\tau}(h)$ denote the stage game matching in the $\tau$-th period of $h$. Following every $t$-period (ex ante or ex post) history $\hat{h} \in \fhistories \cup \histories$, let
\begin{equation*}
U_f(\hat{h}|\mu) \equiv (1-\delta) \exp_{\mu}\Big[  \sum_{\tau = t}^\infty \delta^{\tau - t} u_f ( m_\tau(h ) ) \, \Big| \, \hat{h} \, \Big].
\end{equation*}
denote the continuation payoff hospital $f$ obtains from $\mu$ following $\hat{h}$, where the expectation is taken with respect to the measure over $\fhistories_\infty$ induced by $\mu$ conditional on $\hat{h}$.

\paragraph{{Deviation Plan}.} In repeated matching markets, unlike in static matching markets, players may participate in deviations that span multiple periods. I make two observations that allow us to tractably analyze these deviations.

Recall that there are three kinds of deviations in the stage game. Our first observation is that a deviation in the stage-game by the singleton coalition $\{f\}$, where $f$ fires a subset of its employees, is equivalent to a deviation by the coalition consisting of $f$ and the students that remain with $f$. It is therefore without loss to focus in the stage game only on (i) deviations by coalitions in the form of $\{w\}$, and (ii) deviations by coalitions in the form of $\{f\}\cup W$.

The second observation is that as a long-lived player, a hospital $f$ can participate in a sequence of deviations, by forming coalitions with students across multiple generations. Each of these coalitions must be immediately profitable for the participating short-lived students, but not necessarily for $f$ since hospitals care about the utility it collects from the entire sequence. 

Motivated by this second observation, I define a \term{deviation plan} for hospital $f$ as a complete contingent plan that specifies, at each ex post history, a set of students with whom $f$ wishes to form a deviating coalition: a {deviation plan} for hospital is a mapping $d_f: \histories \rightarrow  2^{\students}$. Together with the original matching process, a deviation plan generates an altered distribution over the outcome of the game $\fhistories_{\infty}$. Given a matching process $\mu$ and $f$'s deviation plan $d_{f}$, the \term{manipulated matching process}, denoted $[\mu, (f, d_{f})]: \histories \rightarrow \smatchings$, is a matching process defined by:
\begin{equation*}
	[\mu, (f, d_{f})] (h) \equiv \big [\mu(h), (f, d_{f}(h)) \big] \;\;\; \text{for every }\; h\in \histories.	
\end{equation*}
The deviation plan $d_{f}$ is \term{feasible} if $|d_f(h) \le q_f|$ and $f \succ_w \mu(w|h)$ for every ex post history $h$ and $w \in d(h)\backslash \mu(f|h)$: at every history,  any student specified by the deviation plan is either already recommended by $\mu$ to work for $f$, or finds himself strictly better off to work for $f$ against the recommendation by $\mu$. Lastly, $d_f$ is \term{profitable} if there exists ex post history $h$ such that $U_{f} (h|[\mu, (f, d_{f}) ])\allowbreak > U_{f}(h|\mu)$.

For much of the analysis in this paper, I will make use of a special class of deviation plans:  a deviation plan $d_{f}$ is a \term{one-shot deviation} from a matching process $\mu$ if there is a unique ex post history $\hat{h}$ where $d_{f} (\hat{h}) \ne \mu (f|\hat{h})$. 





\paragraph{Self-Enforcing Matching Process.}
The two observations above motivate the solution concept.
\begin{definition} \label{definition:self-enforcing-matching-process}
	A matching process $\mu: \histories \rightarrow \smatchings$ is \term{self-enforcing} if
	\begin{enumerate}
		\item $\mu(w|h) \succeq_w  w$ for every $w$ at every ex post history $h$; and
		\item no hospital $f$ has a deviation plan that is both feasible and profitable.
	\end{enumerate}
\end{definition}

The first requirement in \cref{definition:self-enforcing-matching-process} guards against deviations by singleton students in every generation; the second requirement asks that no hospital can profit from chaining together a sequence of deviating coalitions, each one of which immediately profitable for the deviating students, but not necessarily for the hospital itself. Note that these requirements are imposed at every ex post history, including those that are off path: this embeds sequential rationality in the same way as subgame perfection does in a repeated non-cooperative game.

In one-shot matching market, \cref{definition:self-enforcing-matching-process} coincides with the definition for static stable matchings. Unlike the core, however, \cref{definition:self-enforcing-matching-process} does not allow for deviations that may comprise multiple hospitals. As argued in \cite{roth1991natural}, deviations consisting of multiple hospitals are likely difficult to coordinate, so coalitions with only one hospital are most relevant when modeling the stability of hospital resident matching. In \cite{aliliu} we show that in the absence of transfers, modeling coalitions consisting of multiple long-run players (or hospitals in the current context) does not alter the sustainable outcomes when players are patient.
\bigskip

 \cref{lemma: one-shot-deviation} below establishes a one-shot deviation principle for self-enforcing matching processes: to check for the stability of matching process, instead of checking all deviation plans, it suffices to focus on those that only depart from the matching process at a single history.

\begin{lemma}[One-Shot Deviation Principle] \label{lemma: one-shot-deviation}
	A matching process $\mu$ is self-enforcing if and only if
	\begin{enumerate}
		\item $\mu(w|h) \succeq_w  w$ for every $w$ at every ex post history $h$; and
		\item no hospital $f$ has a one-shot deviation that is both profitable and feasible.
	\end{enumerate}
\end{lemma}

The proof of \cref{lemma: one-shot-deviation} follows similar arguments to the one-shot deviation principle for repeated non-cooperative games, and is relegated to  \cref{section: supp prelim}. An immediate implication of \cref{lemma: one-shot-deviation} is that self-enforcing matching processes exist at every patience level.

\begin{observation}\label{theorem:existence}
	There exists a self-enforcing matching process for every $0\le \delta<1$.
\end{observation}
Recall that we assumed hospitals preferences are responsive. Standard results in the static matching literature (see, for example, Theorem 6.5 in \citealp{rothsotomayor1992}) then ensure the existence of stable static matchings; By \cref{lemma: one-shot-deviation}, the infinite repetition of a static stable matching is always a self-enforcing matching process.

\subsection{Top Coalition Sequence} \label{subsection: top coalition sequence}


In the second example in \cref{section: example}, history dependence is powerless to top hospitals due to their immunity to future punishments. In a general matching environment without assuming common student preferences, the appropriate notion of top players is captured by \textit{top coalitions}.

Fix an arbitrary subset of hospitals and students $F\cup W \subseteq \hospitals\cup \students$. A hospital $\hat{f} \in F$ and a set of students $\hat{W} \subseteq W$, $|{\hat{W}}| \le q_f$ form a \term{top coalition} in $F\cup W$ if
\begin{enumerate}
	\item $\tilde{u}_{\hat{f}}\,(\hat{W} ) \ge u_{\hat{f}}\,( W' ) \text{ for all } W' \subseteq W$: $\hat{W}$ is $\hat{f}$'s favorite group of students among $W$; and
	\item $\hat{f} \succeq_w f'$ for all  $w\in \hat{W}$ and $f'\in F\cup\{w\}$: $\hat{f}$ is the favorite hospital for every student in $\hat{W}$.
\end{enumerate}

In other words, $\hat{f}$ and $\hat{W}$ are mutual favorites. The  \term{top coalition sequence} takes this idea further by iteratively finding and eliminating top coalitions in the remaining players, until no new top coalition can be found. 

\begin{definition} \label{definition:top-coalition-sequence} 
The top coalition sequence is the ordered set $\tcs = \{(\hat{f}_1, \hat{W}_{1}), (\hat{f}_2, \hat{W}_{2}), \ldots \}$ produced by the following procedure:\footnote{Whenever it causes no confusion, I will use $\tcs$ to denote both the set of $(f,W)$ pairs and the set of players that show up in those pairs.}
\begin{itemize}
	\item Initialization: Set $\tcs = \emptyset$;
	\item New Phase: 
	\vspace{-6pt}
	\begin{enumerate}
		\item If $(\hospitals \cup \students) \backslash \tcs$ contains no top coalition, stop.
		\item If $(\hospitals \cup \students) \backslash \tcs$ has a top coalition $(\hat{f}, \hat{W})$, add $  (\hat{f}, \hat{W} ) $ to $\tcs$ and restart New Phase.

	\end{enumerate}
\end{itemize}
\end{definition}


Top coalition sequence is related to but distinct from the ``top coalition property'' studied in a variety of cooperative game settings. While the top coalition property is an assumption which requires certain subsets of players to have at least one a top coalition, the top coalition sequence, by contrast, is an object that is constructed from arbitrary player preferences.\footnote{See, for example, \cite{eckhout2000}, \cite{banerjeekonishisonmez}, \cite{pycia2012stability} and \cite{wu2015} for applications of the top coalition property. See \cite{peralta2019} for an example of the application of top coalition sequence in static matching environments.}

Below are some observations on how the composition of the top coalition sequence may depend on the preference configurations in the market. 

\begin{observation} \label{observation:common-ranking}
Suppose all hospitals are acceptable to all students ($f \succ_w w$ for all $w\in \students$):
\begin{enumerate}
	\item If students share a common preference ranking over hospitals, then all players are in the top coalition sequence.
	\item When hospitals share a common utility function over students, the top coalition sequence may be empty.
\end{enumerate}
\end{observation}
The first observation above follows from iterative elimination of the top hospital along the students' shared preference list, just like in the second example in \cref{section: example}. The second observation is illustrated through the following example.
\begin{example}
	$\hospitals= \{f_1, f_2\}$ and $\students = \{w_1, w_2, w_3, w_4\}$. The two hospitals $f_1$ and $f_2$ are identical: both have capacity $q_{f_1}=q_{f_2}=2$ and share a common utility function $u_{f_1}(w_i) = u_{f_2}(w_i) = i$. Suppose $\succ_{w_i} = f_1, f_2, w_i $ if $i$ is odd, and $\succ_{w_i} = f_2, f_1, w_i $ if $i$ is even.
	
	The procedure in \cref{definition:top-coalition-sequence} stops at the first step: both hospitals point to $\{w_3, w_4\}$ as their favorite students, but since neither $f_1$ nor $f_2$ is the favorite for both $\{w_3, w_4\}$, there is no top coalition and $\tcs=\emptyset$.
\end{example}

\subsection{The Limit of Self-Enforcement} \label{subsection: repeated results}

The results in this section explore the extent to which history dependence can be used to alter the matches obtained by players in the top coalition sequence versus those that are not.


\paragraph{Impossible to Motivate Top Coalition Sequence.} As \cref{theorem:tcs} shows, if a student is in the coalition sequence, then even in a repeated matching market and when hospitals are patient, it is impossible to send her to a rural hospital without creating instability.

\begin{theorem} \label{theorem:tcs}
Suppose $(\hat{f}, \hat{W})$ is in the top coalition sequence, then $\hat{f}$ is matched to $\hat{W}$ in all static stable matchings. Moreover, for every $0<\delta<1$, $\hat{f}$ is matched to $\hat{W}$ in every self enforcing matching process at every ex post history.
\end{theorem}
\cref{theorem:tcs} states that the hospitals and students in the top coalition sequence are always matched together in both static and repeated matching markets. This holds in a stark sense in repeated matching markets since it applies regardless of hospitals' patience and after every ex post history, including those that are off-path.



Here are the key steps for establishing \cref{theorem:tcs}: first, $(\hat{f}, \hat{W})$ being mutual favorites implies that they must be matched together in any static stable matching. Second, in a repeated matching market,  $(\hat{f}, \hat{W})$ being mutual favorites further implies that $\hat{f}$ is not punishable through continuation value: whenever a matching process recommends $\hat{f}$ to match with $W\ne \hat{W}$, everyone in $\hat{W}$ is willing to deviate with $\hat{f}$, so $\hat{f}$'s continuation value cannot be lower than $u_{\hat{f}}(\hat{W})$; at the same time, $\hat{f}$'s continuation value also cannot be higher than $u_{\hat{f}}(\hat{W})$, so $\hat{f}$'s continuation value must be precisely $u_{\hat{f}}(\hat{W})$, no matter what happens in the current period. Without credible changes in continuation value, $\hat{f}$ behaves just like a short-lived player, so it must always match with $\hat{W}$ at every ex post history. Finally, an inductive argument extends this logic to the entire top coalition sequence.

%
%

\paragraph{Difference from Standard Folk Theorem.} The ``impossibility'' implication from \cref{theorem:tcs} stands in contrast with what one might expect from standard folk theorems for repeated games, where many outcomes are sustainable at high patience levels---here, the matching outcome for players in $\tcs$ is unique no matter how high the  patience is. 


To understand why, think of the hospitals as players who choose groups of students as their ``actions''.\footnote{See the Conclusion and  \cref{section: product structure} for more detailed discussions on this alternative representation of the stage game, and how it relates to cooperative game form of two-sided matching games.} Note that whenever a stage-game matching $m$ is recommended, a hospital $f$ can always deviate from this recommendation by choosing from students in $D_f(m) \equiv m(f) \cup \{w \in \students:  f\succ_w m(f) \}$. At the same time, a self-enforcing matching process can only credibly recommend stage-game matchings in $\smatchings^{\circ}\equiv\{m\in \smatchings: m \succeq_w w \text{ for all } w\in \students\} $.
One way to define the minmax payoff for each hospital $f$ is then
\begin{equation} \label{equation: wrong-minmax}
	\underline{u}_{f} \equiv \min_{m \in \smatchings^{\circ}} \max_{ W \subseteq D_f(m) } u_f(W).	
\end{equation}
A folk theorem would state that any payoff profile giving every hospital $f$ strictly higher than its respective $\underline{u}_{f}$ can be sustained through a self-enforcing matching process as $\delta \rightarrow 1$. 

However, the problem with this approach is that whenever $\tcs$ is nonempty and contains, say, $(\hat{f}, \hat{W})$ as its first element, then $\hat{W} \subseteq D_{\hat{f}}(m)$ for all $m \in \smatchings^{\circ}$. According to \cref{equation: wrong-minmax}, hospital $\hat{f}$'s minmax payoff satisfies 
\begin{equation*} 
	\underline{u}_{\hat{f}} = \min_{m \in \smatchings^{\circ}} \max_{ W \subseteq D_f(m) } u_f(W) \ge u_{\hat{f}}(\hat{W}) = \max_{W \subseteq \students } u_{\hat{f}}(W).	
\end{equation*}
In other words, $\hat{f}$'s minmax payoff is identical to its highest feasible payoff, so the set of feasible payoff profiles giving $\hat{f}$ \textit{strictly} higher than $\underline{u}_{\hat{f}}$ is an empty set. A top coalition hospital essentially has an ``action'' that guarantees itself the highest possible payoff from the stage game independent of the ``actions'' of other hospitals, and the folk theorem is always vacuous for such payoff structures.


Moreover, \cref{equation: wrong-minmax} also leads to incorrect minmax payoffs for hospitals that are not in $\tcs$. In light of \  \cref{theorem:tcs}, some stage-game matchings in $\smatchings^{\circ}$ fails to account for the constraint that any credible recommendation must match $\hat{f}$ with $\hat{W}$. As a result, \cref{equation: wrong-minmax} incorrectly assumes that $\hat{f}$'s hiring capacity can be used when punishing other hospitals, which underestimates the minmax payoffs for hospitals other than $\hat{f}$.

It is worth noting that the subtlety introduced by the top coalition sequence is different from the failure of full dimensionality that is studied in \cite{wen1994folk} and \cite{FLT}. In \cite{wen1994folk} and \cite{FLT}, there is a non-empty set of payoff profiles that are feasible and give each player strictly higher than her minmax; however, this set lacks dimensionality due to long-run players having aligned preferences. The presence of top coalition sequence, by contrast, 
leads to an empty set of such payoff profiles. The construction of top coalition sequence in \cref{definition:top-coalition-sequence} is an interative process of finding degenerate payoff dimensions while recalibrating minmax payoff for remaining hospitals, until we arrive at a reduced game without new top coalitions.


\paragraph{Modified Folk Theorem in the Reduced Game.} Motivated by the discussion above, let us introduce a few notations that are useful for analyzing the reduced game after the top coalition sequence has been romoved. Let $\reduced \equiv (\hospitals\cup \students)\backslash \tcs$ denote the players that are not in the top coalition sequence, and $\redmatching$ denote the set of stage-game matchings that ensure the top coalition sequence is matched together:
\begin{equation*}
	\redmatching \equiv \big\{m\in \smatchings: m(\hat{f}) = \hat{W} \text{ for all } (\hat{f},\hat{W} ) \in \tcs \big\}.
\end{equation*}
Let $\redmatching ^{\circ}$ denote the stage-game matchings in $\redmatching$ that are individually rational for the students. 
For every hospital $f\in \hospitals \cap \reduced$, its \textit{reduced-game minmax} is 
\begin{equation} \label{equation:reduced-minmax}
	\underline{u}^{\reduced}_{f} \equiv \min_{m \in \redmatching^{\circ}} \max_{ W \subseteq D_f(m) } u_f(W).	
\end{equation}
Notice that for every hospital $f\in \hospitals \cap \reduced$, its reduced-game minmax is higher than the naive calculation in \cref{equation: wrong-minmax}, since the minimization is taken over a more restricted set of stage-game recommendations. Finally, let $\Lambda^* \equiv \{\lambda \in \Delta(\redmatching^{\circ}): u_f({\lambda}) > \underline{u}^{\reduced}_f \text{ for all } f \in \hospitals\cap\reduced \}$ denote the randomizations over $\redmatching^{\circ}$ that secure each hospital strictly higher than its reduced-game minmax.


In contrast to \cref{theorem:tcs}, \cref{theorem:folk-theorem} shows that history dependence can be used to change the matches obtained by players outside of the top coalition sequence: every random matching in  $\Lambda^*$ can be sustained on path in a stationary manner in a self-enforcing matching process.

\begin{theorem}\label{theorem:folk-theorem}
For every $\lambda \in \Lambda^*$, there is a $\underline{\delta}$ such that for every $\delta\in(\underline{\delta},1)$, there exists a self-enforcing matching process that randomizes according to $\lambda$ in every period. 
\end{theorem}

The first step of the proof is showing that hospitals' payoffs in the reduced game always satisfy the Non-Equivalent Utilities (NEU) condition: no hospital's payoff can be a positive affine transformation of another \citep{abreuduttasmith}.  The proof then uses this condition to construct player-specific punishments to deter deviations. Given these punishments, the final step adapts the construction from \cite{fudenberg1986folk} to show that the payoff profile corresponding to any $\lambda \in \Lambda^*$ can be sustained in a self-enforcing matching process when hospitals are patient.


%

\section{Large Market Analysis} \label{section: large markets}

We see in \cref{section: repeated market} that certain hospitals are untouchable, while the others can be motivated through history dependence. But how big is the untouchable set, relative to those that are adjustable? The goal of this section is to quantify asymptotically the size of these two kinds of hospitals. In order to do this, I build on the repeated matching model introduced in \cref{section: repeated market}, but augment it with randomly drawn students and focus on large market analysis. \cref{subsection: lm setup} introduces this setup; \cref{subsection: lm asymp} characterizes the large market asymptotics in this environment.

\subsection{The Setup} \label{subsection: lm setup}
I consider a sequence of market sizes $n$, letting $n$ diverge to infinity. In a market of size $n$, the stage game consists of $n$ hospitals $\hospitals_n$, each with hiring quota $q$,\footnote{The assumption that each hospital has identical quota is only for convenience. The same results continue to hold if each hospital has a different quota.} and students $\students_n$, where $\abs{\students_n} = \lceil \beta nq \rceil$ with $\beta>0$. Let $\smatchings_n$ denote the set of stage-game matchings among $\hospitals_n$ and $\students_n$.



I use the finite-tier random preference model to capture positive preference correlations that arise from quality differentiations in the market.\footnote{See, for example, \cite{ashlagi2017} for an ordinal version of the multiple-tiered preferences, and \cite{che2019efficiency} for a cardinal parameterization of this class of preferences. See also \cite{lee2016} for a different way of modeling preference correlation without ``tiers.''} For every $n$, hospitals can be partitioned into $K$ quality classes $\hospitals_n = \{ \hospitals_n^1, \hospitals_n^2, \ldots, \hospitals_n^K \}$. Every student prefers a hospital from a higher quality class to those from a lower quality class; but each student's preference ranking over hospitals within the same quality class $\hospitals_n^k$ is drawn uniformly from all permutations of $\hospitals_n^k$. Let $\pi_n$ denote a realization of student preferences that are compatible with this restriction. I assume that the proportion of tier-$k$ hospitals, $\abs{\hospitals^k_n}/n$, converges to $x_k \ge 0$ for $1\le k \le K$.

Similarly, students can be partitioned into $L$ quality classes $\students_n = \{\students^1_n, \students^2_n, \ldots, \students^L_n \}$. When a hospital $f$ matches with student $w\in \students^l_n $, the hospital receives
\begin{equation*}
	\tilde{u}_f(w) = V(C_l, \zeta_{f,w} ),
\end{equation*}
where $C_l$ is the common value shared by all students in $\students_n^l$ satisfying $C_l > C_{l'}$ for all $l<l'$, and $\zeta_{f,w}$ is the idiosyncratic match quality between $f$ and $w$. I assume that the quality component for each tier, $C_l$, is constant over time, while the idiosyncratic components $\zeta_{f,w}$ are drawn independently for every student in each cohort from the uniform distribution over $[0,1]$. $V(.,.)$ is a continuous and strictly increasing function from $\Re^2_+$ to $\Re_+$, and satisfies $V(C_l, 0) > V(C_{l'},1)$ for all $l < l'$ (so there is no overlap between tiers). Hospitals have additive utilities for each hiring slot: $\tilde{u}_f(W) = \sum_{w\in W} \tilde{u}_f(w) $, and derive zero utility from unfilled positions. Let $ \zeta_n =  \{\zeta_{f,w} \}_{f\in \hospitals_n, w\in \students_n}$ denote a realization of the matrix of idiosyncratic match qualities. I assume that the proportion of tier-$l$ students, $\abs{\students^l_n}/ \abs{\students_n} $, converges to $y_l \ge 0$ for $1\le l \le L$.

The timing in each period is as follows: first, a new cohort of students arrive and preferences $\pi_n$ and $\zeta_n$ are realized; the public randomization $\omega \in \Omega$ is then realized; based on the realization of $(\pi_n, \zeta_n, \omega_n)$, a stage-game matching is recommended for $\hospitals_n$ and $\students_n$; players then decide whether to deviate from this recommendation, which determines the outcome of the stage-game. I will refer to the realization of $s_n = (\pi_n, \zeta_n, \omega_n ) \in S_n$ together as a \term{state}. The notions of ex ante and ex post histories, introduced in \cref{section: repeated market}, are modified accordingly with $s_n$ replacing $\omega$.



\subsection{Which Hospitals Are Untouchable?} \label{subsection: lm asymp}
I find that in repeated large matching markets, the vast majority of hospitals can be motivated dynamically to voluntarily reduced their capacities, making more residents available to hospitals that would otherwise struggle to fill their hiring quotas. The only exceptions, should they exist, are elite hospitals that by definition only makes up a vanishingly small fraction of the market.

Specifically, I say the top class of hospitals $\hospitals^1_n$ are elite hospitals if their size is vanishing relative to the top class of students.

\begin{definition}
	$\hospitals^1_n$ is an elite quality class if $\abs{\hospitals^1_n}/\abs{\students^1_n} \rightarrow 0$ as $n\rightarrow \infty$.
\end{definition}

A hospital in $\hospitals^1_n$ faces no vertical competition, since it dominates hospitals from all other lower quality classes. If $\hospitals^1_n$ is in addition an elite hospital class, then it also faces diminishing horizontal competition from within the same quality class. Note that elite hospitals may not necessarily exist. For example, if $\lim_{n\rightarrow\infty} \hospitals^1_n/n >0 $, then no hospitals in the market would qualify as elite hospitals.

\cref{theorem: cannot reduced capacity} shows that no matter how patient the hospitals are, when market size is large, it would be difficult to reduce elite hospitals' hiring capacities through history dependence.

\begin{theorem} \label{theorem: cannot reduced capacity}
Suppose $\hospitals^1_n$ is an elite hospital class, then for every discount factor $0<\delta < 1 $, there exists $N$ such that for all $n \ge N$, every self-enforcing matching process $\mu$ satisfies $ \abs{\mu(f|h)} = q$ for all $f\in \hospitals^1_n$ and all ex post history $h$.
\end{theorem}

Note that \cref{theorem: cannot reduced capacity} is not a folk theorem: instead of taking patience $\delta$ to $1$, I consider an arbitrary fixed $\delta$ while taking market size $n$ to infinity. 
Here is the intuition for establishing \cref{theorem: cannot reduced capacity}. As $n \rightarrow \infty$, an elite hospital $f$ faces no vertical competition and diminishing horizontal competition. When market is sufficiently large, $f$ can with very high probability fill all its hiring slots with students who are mutual favorites with it. This guarantees $f$ \textit{almost} the highest possible continuation value, no matter what happens in the present. Even though this is not exactly the highest possible continuation value, since $\delta$ is fixed instead of arbitrarily close to $1$, such small variation in future value is not enough to motivate hospital $f$ to give up a hiring slot in the current period. As a result, $f$ always hires at full capacity after every history.

Note that for every \textit{fixed} $N$, there is a sufficiently high $\delta$ that allows self-enforcing matching processes with reduced capacity for $\hospitals^1_n$. So another way of interpreting \cref{theorem: cannot reduced capacity} is that the ability to reduce $\hospitals^1_n$'s capacity comes down to a race between patience $\delta$ and market size $N$, and no amount of patience can be sufficient for all large market sizes.




While \cref{theorem: cannot reduced capacity} highlights the difficulties in reducing the capacity of elite hospitals, it is worth remembering that elite hospitals, by definition, only make up a vanishing fraction of the total seats in the market. To the extent that the policy interests lie in allocating more residents to rural hospitals, their impact is negligible.
The next result, \cref{theorem: can reduce capacity}, confirms that as long as a hospital class makes up a non-vanishing fraction of the market, it is indeed possible to reduce their capacities. Perhaps surprisingly, this also applies to $\hospitals^1_n$. Furthermore, unlike \cref{theorem: cannot reduced capacity}, the affirmative message of \cref{theorem: can reduce capacity} is not the result of a race between $\delta$ and $N$: \cref{theorem: can reduce capacity} holds true especially when both $\delta$ and $N$ are large.

\begin{theorem} \label{theorem: can reduce capacity}
Suppose $\lim_{n\rightarrow \infty} \abs{\hospitals^k_n}/n >0$ for some $1\le  k \le K$. We can find a sequence of matching processes $\{\mu^*_n\}_{n=1}^\infty$, a discount factor $0< \underline \delta <1$, a market size $N$, and probability $p^*>0$, such that for all $\delta \in (\underline{\delta},1 )$ and $n>N$:
\begin{enumerate}
	\item $\mu^*_n$ is a self-enforcing matching process; and
	\item $P(|\mu^*_n(f|h)| <q) \ge p^*$ for every $f\in \hospitals^k_n$ and every on-path ex post history $h$.
\end{enumerate} 
\end{theorem}


The second condition in \cref{theorem: can reduce capacity} says that the matching process $\mu^*_n$ can reduce the capacity of all hospitals in $\hospitals^k_n$ with positive probability; moreover, this positive probability is bounded away from $0$ as $n\rightarrow \infty$.

The proof of \cref{theorem: can reduce capacity} builds on techniques from the repeated game literature and the large matching market literature. In order to sustain matching processes with reduced capacities for $\hospitals^k_n$, I utilize a folk-theorem construction like that from \cite{fudenberg1986folk}. On path, hospitals in $\hospitals^k_n$ are asked to reduce their capacity and as a reward for their compliance, they are sometimes allowed to choose their favorite students through a hospital-proposing random serial dictatorship. I use existing results on large matching markets from \cite{che_tercieux_effequiv} to characterize the expected payoffs from such rewards. In case of a non-compliance from any hospital $f\in \hospitals^1_n$, $f$ undergoes a punishment phase for a few periods to offset the gains from its deviation, before the market resets with $f$ being asked to make an even greater reduction in its capacity.

However, unlike standard repeated games, large matching markets also present a unique challenge. In matching markets with correlated preferences, hospitals may obtain close to efficient payoffs from all stable matching algorithms \citep{pittel1989, lee2016, ashlagi2017}. The crucial part of the proof for \cref{theorem: can reduce capacity} boils down to finding a ``punishment algorithm'' that: 1. holds down a hospital's expected payoff asymptotically as market size gets large, and 2. makes sure the punished hospital can find no student willing to jointly deviate from this punishment---such an algorithm essentially plays a role analogous to that of a minmax action profile in standard repeated games. I find that certain variants of the student-proposing serial dictatorship satisfy both requirements. These punishment algorithms are effective even when all static stable matchings are close to utilitarian-efficient.

The formal proofs are in \cref{section: proof of thm4} and \cref{section: supp matching process}. I describe the intuition for the punishment algorithm here. For simplicity, let us focus on the case when matching is one-to-one, and suppose that $f\in \hospitals^1_n $ is the hospital being punished.
\smallskip

\noindent \textit{\underline{The Algorithm}:}
The algorithm assigns priorities to students according to $f$'s preference ranking, and then run the student-proposing serial dictatorship based on assigned priorities. As students propose and exit the market, the algorithm moves down $f$'s preference list. How high $f$ ranks its match therefore boils down to how early it is drafted by a student.
\smallskip

\noindent \textit{\underline{Why the Algorithm Lowers $f$'s Payoff}:} From the students' perspective, the hospitals are being sampled without replacement, where a new hospitals is picked from the remaining ones each round. Since preferences within each quality class are uniformly random, the number of draws it takes for $f$ to be sampled is uniformly distributed from 1 to $\abs{\hospitals^1_n}$, so $f$ will be matched uniformly within the top $|\hospitals^1_n|/n$ fraction of its preference list. If $|\hospitals^1_n|/n$ is non-vanishing, there is a non-vanishing probability that $f$ will be excluded from the very top of its preference list, therefore lowering $f$'s payoff. In this way, the algorithm essentially leverages horizontal competition (or in other words, the size of $|\hospitals^1_n|$) to punish $f$.
\smallskip


\noindent \textit{\underline{Why $f$ Can Find No Student to Deviate With}:} If a student did not get matched to $f$, there are two possibilities: either she picked a better hospital early on, or she had low priority because she is ranked lower than the student who did get matched with $f$. Either way, $f$ cannot find any better student to replace its match.

\section{Conclusion} \label{section: conclusion}

This paper provides a framework and solution concept for studying stability in repeated matching markets, which combine elements from repeated noncooperative games with the cooperative stability notion for two-sided matching markets. While history-dependent play can be used to alter the matches obtained by some hospitals, some other hospitals are immune to such influences due to the unique payoff structure in matching environments. To understand the impact of such hospitals, I consider large matching markets with correlated preferences. I find that in large matching markets, these elite hospitals make up at most a negligible fraction of the market.

The analysis of repeated matching markets provides insights into market design. In one-shot matching markets, the Rural Hospital Theorem rules out the possibility of improving matching outcomes for rural hospitals. A contribution of the current paper is to show that in repeated matching markets, while it is difficult to impose hiring caps on elite hospitals, it is nevertheless possible to reduce the hiring capacity of most, if not all, other hospitals in the market without violating stability. Stability therefore should not represent a hurdle for designing matching processes that can make more residents available for rural hospitals.
A promising avenue for future research is in understanding whether  matching processes like simple exclusion can be self-enforcing in an approximate sense in large matching markets.

By directly modeling deviations by coalitions, the repeated matching model that I developed in this paper preserves the cooperative spirit of static matching models. My analysis is also consistent with an alternative model where hospitals treat subsets of students as actions, and student preferences are used to compute the payoffs accruing to hospitals under each action profile. In this strategic-form game, a hospital in a top coalition has an action that guarantees its highest payoff regardless of the actions of other hospitals. This equivalence holds only when utility is non-transferable and there are no payoff externalities (see \cref{section: product structure} for a more detailed discussion and an illustrative example). In ongoing work, I consider repeated matching markets in such environments, and prove the existence of self-enforcing matching process under a broad class of preferences.

{\small
	\addcontentsline{toc}{section}{References}
	\setlength{\bibsep}{0.25\baselineskip}
	\bibliographystyle{aer}
	\bibliography{DynMatch}
}

\addtocontents{toc}{\protect\setcounter{tocdepth}{1}}

\appendix
\renewcommand{\theequation}{\arabic{equation}}
  \renewcommand{\thesection}{\Alph{section}}
\small

\section{Appendix} \label{section: appendix}


\subsection{Proof of \cref{theorem:tcs}}

Suppose $\tcs = \{(\hat{f}_k, \hat{W}_k)   \}_{k=1}^K$. In both one-shot and repeated matching environments, the proof proceeds by induction.

\noindent \textbf{Static Stable Matching:} In every static stable matching $\hat{f}_1$ and $\hat{W}_1$ must be matched together since they are mutual favorites. Suppose $\hat{f}_i$ and $\hat{W}_i$ must be matched together in all stable matchings for $1 \le i \le k-1$, but suppose by contradiction that there is a stable matching $m$ where $\hat{f}_k$ and $\hat{W}_k$ are not matched together. By the induction hypothesis, $m(\hat{f}_k) \subseteq \students \backslash \cup_{i=1}^{k-1}\hat{W}_i $ so $\tilde{u}_{\hat{f}}(\hat{W}_k) > \tilde{u}_{\hat{f}}(m(\hat{f}_k))$, and $m(w) \in \hospitals \backslash \{\hat{f}_i: 1 \le i \le k-1 \}$ so $\hat{f}_k  \succ_w  m(w)$ for all $w \in \hat{W}_k\backslash m(\hat{f}_k) $. This is a contradiction to $m$ being a stable matching, since $\hat{f}_k$ and $\hat{W}_k$ find it profitable to jointly deviate. So $\hat{f}_k$ and $\hat{W}_k$ must be matched together in all stable matchings. This completes the induction step.

\noindent \textbf{Self-Enforcing Matching Process:} The proof again proceeds by induction. First, I prove that in every self-enforcing matching process $\mu$, ${\mu}( \hat{f}_1|h ) =\hat{W}_1$ for all $h\in \histories^F$. Suppose by contradiction that ${\mu}(\hat{f}_1 |\tilde{h}) \ne \hat{W}_1$ at some $\tilde{h} \in \histories$. Consider the deviation plan $d_1$ defined by $d_1(h) =\hat{W}_1 $ for all $h\in \histories$. $d_1$ is clearly feasible for $\hat{f}_1$ since $\hat{f}_1$ is every student's favorite hospital. In addition,
	\begin{equation*} 
	\begin{split}
		U_{\hat{f}_1}({\tilde{h}|\mu } ) & = (1-\delta) \tilde{u}_{\hat{f}_1}(\mu( \hat{f}_1 |\tilde{h})) + \delta U_{\hat{f}_1}({\tilde{h}, \mu(\tilde{h})}|\mu ) \\
		& < (1 -\delta) \tilde{u}_{\hat{f}_1} (\hat{W}_1) + \delta u_{\hat{f}_1} (\hat{W}_1) = U_{\hat{f}_1}\big ({\tilde{h} \big| [\mu, (\hat{f}_1, d_1) }] \big)
	\end{split}
	\end{equation*}
so $d_1$ is also profitable for $\hat{f}_1$. This is a contradiction to $\mu$ being self-enforcing. So ${\mu}(\hat{f}_1|h ) =\hat{W}_1$ for all $h\in \histories$. 
	
Suppose it has been shown that in every self-enforcing matching process $\mu$, ${\mu}(\hat{f}_i| h )  = \hat{W}_i $ for $i = 1, \ldots, k-1$ at every ex post history $h\in \histories$, and suppose by contradiction that there is a self-enforcing matching process $\mu$ such that ${\mu}(\hat{f}_k|\tilde{h} )  \ne \hat{W}_k $ for some $\tilde{h} \in \histories$. By the inductive hypothesis, $\mu(\hat{f}_k| h ) \subseteq \students \backslash \cup_{i=1}^{k-1} \hat{W}_i$ at all $h\in \histories$, so $\tilde{u}_{\hat{f}_k}( {\mu}(\hat{f}_k|\tilde{h} ) ) < \tilde{u}_{\hat{f}_k} (\hat{W}_k)$, and $U_{\hat{f}_k}({\tilde{h}\, |\, \mu } ) \le \tilde{u}_{\hat{f}_k} (\hat{W}_k) $.

Consider the deviation plan $d_k$ defined by $d_k(h) =\hat{W}_k $ for all $h\in \histories$. $d_k$ is feasible for $\hat{f}_k$ since for every student in $\students \backslash \cup_{i=1}^{k-1} \hat{W}_i$, $\hat{f}_k$ is the best hospital among $\hospitals\backslash \{f_i \}_{i=1}^{k-1}$. In addition
	\begin{equation*} 
	\begin{split}
		U_{\hat{f}_k}(\tilde{h}| \mu  ) & = (1-\delta) \tilde{u}_{\hat{f}_k}(\mu( \hat{f}_k |\tilde{h})) + \delta U_{\hat{f}_k}({\tilde{h}, \mu(\tilde{h})\, |\, \mu } ) \\
		& < (1 -\delta) \tilde{u}_{\hat{f}_k} (\hat{W}_k) + \delta \tilde{u}_{\hat{f}_k} (\hat{W}_k) = U_{\hat{f}_k}\big ({\tilde{h} \big| [\mu, (\hat{f}_k, d) }] \big)
	\end{split}
	\end{equation*}
so $d_k$ is both feasible and profitable, contradicting the assumption that $\mu$ is self-enforcing. So in every self-enforcing matching process $\mu$, ${\mu}(\hat{f}_k| h )  = \hat{W}_k $ at every ex post history $h\in \histories$. This completes the induction.


\subsection{Proof of \cref{theorem: cannot reduced capacity}}

\subsubsection{Preliminaries}
I first establish a few preliminary results in order to prove \cref{theorem: cannot reduced capacity}. \cref{lemma: minmax set full1} proves that when market size is sufficiently large, an elite hospital $f$ is very likely able to fill its quotas with students who: 1. rank $f$ as their favorite hospital, and 2. give $f$ close to the highest possible stage-game utility.
\medskip

For every $f\in \hospitals^1_n$, let $\hat{W}^1_n(f,\epsilon) \equiv \{w\in \students^1_n : f \succeq_w f' \text{ for all } f'\in \hospitals^1_n, \text{ and } u_f (w) > V(C_1 ,1) -\epsilon \}$.

\begin{lemma} \label{lemma: minmax set full1}
Suppose $\abs{\hospitals^1_n}/ \abs{\students^1_n} \rightarrow 0$. In the stage game, for every $\epsilon>0$, there exists $N$ such that $P \big( \big| \hat{W}(f,\epsilon,k ) \big| > q \big) > 1-\epsilon $ for all $n>N$ and every $f\in \hospitals^1_n$. 
\end{lemma}

\begin{proof}
Let $\underline{\zeta} \in [0,1) $ be a number such that $V({C}_1, \underline{\zeta} ) > V(C_1, 1) - \epsilon$. Define $\tilde{W}^1_n(f,\epsilon) \equiv \{w \in \students^1_n : f \succeq_w f' \text{ for all } f'\in \hospitals^1_n, \text{ and } \zeta_{f,w} > \underline{\zeta}  \}$, so $\tilde{W}^1_n(f,\epsilon) \subseteq \hat{W}^1_n(f,\epsilon)$.

From the perspective of hospital $f\in \hospitals^1_n$, every student $w \in \students^1_n$ satisfies $\zeta_{f,w} > \underline{\zeta}$ with probability $1-\underline{\zeta}  >0 $. Furthermore, a student is equally likely to rank any hospital within $\hospitals^1_n$ as her top choice within $\hospitals^1_n$. Let $Y_w(f,\epsilon) $ be the Bernoulli random variable that takes value $1$ if $w \in \tilde{W}^1_n(f,\epsilon)$ and $0$ otherwise. Let $\phi\equiv (1-\underline \zeta)/\abs{\hospitals^1_n} $. Note that the random variables $\{ Y_w(f,\epsilon): w\in \students^1_n \}$ are independently and identically distributed with rate $\phi$, and as a result, $\size{\tilde{W}^1_n(f,\epsilon)} = \sum_{w\in \students_n} Y_w(f,\epsilon)$ follows binomial distribution $B(\abs{\students^1_n}, \phi )$. By the Chernoff bound,


\begin{equation*}
	P\Big ( \,   \big|{\tilde{W}^1_n(f,\epsilon)} \big|   \le q \Big) \; \le  \min_{t>0} e^{tq}  \prod_{w\in \students^1_n} \exp \big[ e^{-tY_w(f,\epsilon )} \big]  
\end{equation*}
Since $\exp \big[ e^{-tY_w(f,\epsilon )} \big]  = 1 + \phi (e^{-t}-1) \le e^{\phi (e^{-t}-1)}$ for all $w\in \students^1_n$, we have
\begin{equation*}
	P\Big ( \,   \big|{\tilde{W}^1_n(f,\epsilon)} \big|   \le q \Big) \; \le  \min_{t>0} e^{tq} \cdot e^{\abs{\students^1_n}\phi (e^{-t}-1)} \le e^{q}\cdot  e^{\abs{\students^1_n}\phi (e^{-1}-1)},
\end{equation*}
where the second inequality above follows from setting $t=1$. Since $\abs{\students^1_n}\phi = (1-\underline{\zeta} )\abs{\students^1_n}/\abs{\hospitals^1_n} \rightarrow\infty$ and $e^{-1}-1<0$, we have $P ( \,   \big|{\tilde{W}^1_n(f,\epsilon)} \big|   \le q ) \rightarrow 0 \text{ as }  n\rightarrow\infty$. Finally, since ${\tilde{W}^1_n(f,\epsilon)} \subseteq {\hat{W}^1_n(f,\epsilon)}$, it follows that $P ( \,   \big|{\hat{W}^1_n(f,\epsilon)} \big|   \le q ) \rightarrow 0 \text{ as }  n\rightarrow\infty$. 
\end{proof}

\cref{lemma: high continuation value1} shows that an elite hospital can secure a high continuation value in large markets.

\begin{lemma} \label{lemma: high continuation value1}
Suppose $\abs{\hospitals^1_n}/ \abs{\students^1_n} \rightarrow 0$. For every $\epsilon>0$ and discount factor $0< \delta <1$, there exists $N$ such that for all $n>N$, the continuation value of every $f\in \hospitals^1_n$ in every self-enforcing matching process $\mu$ satisfies $U_f(\overline{h}|\mu) \ge qV(C_1,1)-\epsilon$ at every ex ante history $\overline{h}$.
\end{lemma}
\begin{proof}

Suppose that the stage game satisfies  $P\big( \big|{\hat{W}^1_n(f,\tilde \epsilon) }\big| >q \big ) > 1 - \tilde \epsilon$ for some $\tilde{\epsilon}$, and consider the following deviation plan $d_f$ by an arbitrary hospital $f\in \hospitals^1_n $: at every ex post history, given the realized preferences in the current stage game,
\begin{equation*}
d_f(h)= \underset{W\subseteq \hat{W}^1_n(f, \tilde \epsilon ),\, \abs{W}\le q }{\arg\max} \;\; \underset{w \in W}{\sum}  u_f(w).
\end{equation*}
I first show that with sufficiently small $\tilde \epsilon$, hospital $f$ can guarantee itself $qV(C_1,1)-\epsilon$ by using $d_f$ to deviate from any matching process. The claim of the lemma then follows by invoking \cref{lemma: minmax set full1}.

Fix an arbitrary matching process $\mu$. Note first that since every student in $\hat{W}^1_n(f,\tilde \epsilon)$ ranks $f$ as their favorite hospital, $d_f$ is a feasible deviation plan for $f$ by construction.

To see that $d_f$ guarantees $qV(C_1,1)-\epsilon$ for sufficiently small $\tilde \epsilon$, let $T$ be large enough so that $\delta^TV(C_1,1)q < \epsilon/2$. At every ex ante history $\overline{h}$, hospital $f$'s continuation payoff from the manipulated matching process $[\mu, (f,d_f)]$ satisfies
\begin{align}
	U_f(\overline{h}\, \big| \, [\mu, (f,d_f)] ) &\ge (1- \delta^T)\left\{ (1-\tilde \epsilon)^T q(V(C_1,1) - \tilde \epsilon)  + [1- (1 - \tilde\epsilon )^T] 0 \right\} + \delta^T \cdot 0 \label{inequality: first T period payoff} \\
	 & = (1- \delta^T) (1-\tilde \epsilon)^T q(V(C_1,1) - \tilde \epsilon)  \nonumber
\end{align}
Inequality \eqref{inequality: first T period payoff} above follows from decomposing $f$'s continuation payoff between those accrued within the first $T$ and those after period $T$; for the first $T$ periods, the expected payoff is further decomposed by whether $| \hat{W}^1_n(f,\tilde\epsilon )|\ge q$ all the way through this phase or not.

Since $(1- \delta^T)(1-\tilde \epsilon)^T q(V(C_1,1) - \tilde \epsilon ) \rightarrow (1- \delta^T)qV(C_1,1)$ as $\tilde \epsilon \rightarrow 0$, we can find $\underline{\epsilon}$ such that $(1-\delta^T)(1- \underline\epsilon)^T q(V(C_1,1) - \underline \epsilon) \geq (1-\delta^T)qV(C_1, 1) - {\epsilon}/2$.
By \cref{lemma: minmax set full1}, there exists $N$ such that $P\big( \big|{\hat{W}^1_n(f,\underline \epsilon) }\big| >q \big ) > 1 - \underline \epsilon$ for all $n>N$. So for all $n>N$, we have
\begin{align*}
	U_f(\overline{h}\, \big| \, [\mu, (f,d_f)] ) & \ge (1-\delta^T)qV(C_1, 1) - {\epsilon}/2 \\
	& = (1-\delta^T)qV(C_1, 1) + \delta^TqV(C_1,1) - {\epsilon}/2 - \delta^TqV(C_1,1)\\
	& > (1-\delta^T)qV(C_1, 1) + \delta^TqV(C_1,1) - {\epsilon}/2 - {\epsilon}/2\\
	& = qV(C_1, 1) -\epsilon
\end{align*}

Finally, if $\mu$ is a self-enforcing matching process, then it must satisfies
\begin{equation*}
	U_f( \overline{h} |\mu ) \ge U_f(\overline{h}\, \big| \, [\mu, (f,d_f)] ) > qV(C_1, 1) -\epsilon
\end{equation*}
at every ex ante history $\overline{h}$.
\end{proof}

\subsubsection{Proof of \cref{theorem: cannot reduced capacity}}
Let $\epsilon \equiv  \frac{1 -\delta }{2\delta}  V(C_1, 0)$. By \cref{lemma: high continuation value1}, there exists $N_1$ such that if $n>N_1$, $U_f(\overline{h}|\mu ) \ge V(C_1, 1)- \epsilon$ for every self-enforcing matching process $\mu$ at every ex ante history $\overline{h}$. In addition, since $\abs{\hospitals^1_n}/ \abs{\students^1_n} \rightarrow 0$, there exist $N_2$ such that $q\abs{\hospitals^1_n} < \abs{\students^1_n}$ for all $n>N_2$. I consider market sizes larger than $N \equiv \max\{N_1, N_2\}$.

Let $\mu$ be an arbitrary self-enforcing matching process, and suppose by contradiction that $|\mu(f|\hat{h})|<q$ at some ex post history $\hat{h}$. I will show that this leads to a contradiction. 

At $\hat{h}$, hospital $f$'s continuation payoff at $h$ satisfies
\begin{equation*}
	U_f(\hat{h}|\mu) \le (1-\delta) \sum_{w\in \mu(f|h)} u_f(w) + \delta qV(C_1,1).
\end{equation*}
Meanwhile, since $|\students^1_n | > q\abs{\hospitals^1_n} $, at history $\hat{h}$, there exists $\hat{w}\in \students^1_n$ such that $\mu(\hat{w} |\hat{h} )\notin {\hospitals^1_n}$. Consider the following deviation plan $d_f$ for hospital $f$:
\begin{equation*}
d_f (h) = 
\begin{cases}
\mu(f|h) & \text{ if } h\ne \hat{h} \\
\mu(f|h)\cup \{\hat{w} \} & \text{ if } h= \hat{h}
\end{cases}
\end{equation*}
Let $m' \equiv \big[\mu(\hat{h} ), \big(f, \mu(f|\hat{h})\cup\{\hat{w} \}  \big)  \big]$ denote the stage-game matching that results from $f$ recruiting $\hat{w}$ in addition to its recommended students at $\hat{h}$. Hospital $f$'s continuation payoff from the manipulated matching process $[\mu,(f,d_f) ] $ satisfies
\begin{align*}
	U_f(\hat{h}|[\mu,(f,d_f) ]) & = (1-\delta) \Big[ \sum_{w\in \mu(f|h)} u_f(w) + u_f(\hat{w} ) \Big] + \delta U_f(\hat{h}, m' | \mu )\\
	& \ge (1-\delta)  \sum_{w\in \mu(f|h)} u_f(w) + (1-\delta ) u_f(\hat{w} )  + \delta [V(C_1,1) -\epsilon ]\\
	& = (1-\delta)  \sum_{w\in \mu(f|h)} u_f(w) + \delta V(C_1,1) + (1-\delta ) u_f(\hat{w} )  -\delta \epsilon \\
	& \ge U_f(\hat{h}|\mu ) + (1-\delta ) u_f(\hat{w} )  -\delta \epsilon
\end{align*}

Since $\hat{w} \in \students^1_n$, $u_f(\hat{w}) \ge V(C_1,0) $. In addition, recall that $\epsilon =  \frac{1 -\delta }{2\delta}  V(C_1, 0)$ by construction. We have
\begin{align*}
	U_f(\hat{h}|[\mu,(f,d_f) ]) 	& \ge U_f(\hat{h}|\mu ) + \frac{1-\delta}{2} V(C_1,0) > U_f(\hat{h}|\mu ),
\end{align*}
contradicting the assumption that $\mu$ is self-enforcing.

\subsection{Proof of \cref{theorem: can reduce capacity}} \label{section: proof of thm4}
I focus on hospital quality classes that fill all their hiring quotas in static stable matchings, since otherwise the statement of the theorem would be trivial. For hospitals that cannot fill their slots in static stable matchings, it is possible to reduce their capacity even further, following similar arguments as the ones below.


\cref{section: submarket} introduces the submarket faced by $\hospitals^k_n$; \cref{section: reward} characterizes the payoffs $\hospitals^k_n$ obtain from the reward matching; \cref{section: punishment} introduces algorithm used to punishment a deviating hospital, and characterize its expected payoff from such punishment. The matching process used to sustain capacity reduction is constructed using these punishments and rewards, I leave the formal construction to \cref{section: supp matching process}.

\subsubsection{Submarkets} \label{section: submarket}
For each $1\le l \le L$, let $Q^{\students}_n (l) \equiv \sum_{l'\le l} |\students^{l'}_n |$ denote the number of students who are in a quality class no worse than $\students^l_n $; similarly, for each $1 \le k \le K$, let $Q^{\hospitals}_n (k) \equiv q \sum_{k'\le k } |\hospitals^{k'}_n |$ denote the number of hospital seats that are in a quality class no worse than $\hospitals^k_n$.

I say that a student quality class $\students^l_n$ is achievable by hospitals in quality class $ \hospitals^k_n$ if
\begin{equation*}
	\lim_{n\rightarrow \infty} \frac{ Q^{\students}_n(l) }{Q^{\hospitals}_n(k-1) } \ge 1 \;\; \text{ and }   \;\; 	\lim_{n\rightarrow \infty} \frac{ Q^{\students}_n(l-1) }{Q^{\hospitals}_n(k) } \le 1.
\end{equation*}
The first inequality above ensures that not all students in $\students^l_n$ can be absorbed by hospitals in higher quality classes than $\hospitals^k_n$; the second inequality ensures that not all seats in $\hospitals^k_n$ can be filled by students in higher quality classes than $\students^{l}_n$. Let $\mathcal{A}(k)\subseteq \{1, \ldots, L\} $ denote the set of student quality classes achievable by hospital quality class $k$. 

In the analysis in this section, when market size is sufficiently large, the only relevant students are those that are in an achievable quality class.  I therefore focus on the submarket that consists of only hospitals in $\hospitals^k_n$ and students in its achievable quality classes.

\bigskip

Let the numbers $\alpha_1, \alpha_2,\ldots,\alpha_I \ge 0 $ be numbers that satisfy $\sum_{i=1}^{I-1}\alpha_i <1$ and $\sum_{i=1}^{I}\alpha_i \ge 1$. Consider a sequence of submarkets: for every $n$, the hospital side is made up of $\hospitals^k_n$, while the student side consists of $\cup_{i=1}^I \mathcal{V}^i_n$, where for each $i$,
\begin{equation*}
\mathcal{V}^i_n \subseteq \students^{i}_n \;\;\;\; \text{ and } \;\;\;\;\frac{\abs{\mathcal{V}^i_n}}{\abs{\hospitals^k_n}q} \rightarrow \alpha_i	 \text{ as $n\rightarrow \infty$.}
\end{equation*}
I treat each seat on the hospital side of the market as an individual player that inherits the preference of its hospital. For each seat $s$, let $\tilde{u}_s(.)$ denote the utility function of the seat, which is identical to $\tilde{u}_f(.)$ of the hospital that it belongs to.

\subsubsection{Reward for Compliance} \label{section: reward}
Let $\hat{\phi}_n$ denote the matching resulting from the seat-proposing random serial dictatorship in the submarket. $\hat{\phi}_n$ is played as a reward for the hospitals when they comply with capacity reduction. \cref{lemma: rsd high payoffs} characterizes the payoff from this reward.
Noting that the matching $\hat{\phi}_n$ is Pareto efficient, I will make use of the following result from \cite{che_tercieux_effequiv}:

\begin{lemma}{\citep{che_tercieux_effequiv}}\label{lemma: CT} 
As $n\rightarrow \infty$,
\begin{equation*}
\frac {\sum_{s} u_s(\hat{\phi}_n(s) )}{\abs{\hospitals^k_n }q} \; \overset {p}{\longrightarrow} \; \sum_{i=1}^{I-1} \alpha_i V(C_i,1) + (1- \sum_{i=1}^{I-1} \alpha_i)V(C_I,1).
\end{equation*}
\end{lemma}

%
%
\cref{lemma: rsd high payoffs} shows that in $\hat{\phi}_n$, every hospital obtains a randomly assigned common value from its matched students, but obtains close to maximum idiosyncratic component from each student.

\begin{lemma} \label{lemma: rsd high payoffs}
For every $\epsilon>0$, there exists $N$ such that 
\begin{equation*}
\exp[u_f(\hat{\phi}_n ) ] > q \big[ \sum_{i=1}^{I-1} \alpha_iV(C_i,1) + (1- \sum_{i=1}^{I-1} \alpha_i)V(C_I,1) \big] -\epsilon.
\end{equation*}
for all market sizes $n>N$ and every $f\in \hospitals^k_n$.
\end{lemma}

\begin{proof}
To prove the claim, it suffices to prove that $\exp[u_s(\hat{\phi}_n ) ] >  \sum_{i=1}^{I-1} \alpha_iV(C_i,1) + (1- \sum_{i=1}^{I-1} \alpha_i)V(C_I,1) - \epsilon$ for all $n>N$ and every individual seat $s$. This is what I will prove below.


Let 
\begin{equation*}	
\overline{u}_n \equiv  \frac{\sum_{s} u_s(\hat{\phi}_n(s))} {\abs{\hospitals^k_n }q}
\end{equation*}
denote the average realized utilities for individual seats, and let $\overline{F}_n$ denote the probability distribution of this random variable $\overline{u}_n$. Let $A_n^- (\epsilon )  \equiv \{ \overline{u}_n \le \sum_{i=1}^{I-1} \alpha_iV(C_i,1) + (1- \sum_{i=1}^{I-1} \alpha_i)V(C_I,1) -\frac{1}{2} \epsilon \}$ be the event that $\overline{u}$ is less than the payoff upper bound, and $A_n^+ (\epsilon ) \equiv \{ \overline{u}_n > \sum_{i=1}^{I-1} \alpha_iV(C_i,1) + (1- \sum_{i=1}^{I-1} \alpha_i)V(C_I,1) -\frac{1}{2} \epsilon \}$ be its complement. 

Note that by the Law of Total Expectation, for every seat $s$,
\begin{align*}
	\exp[u_s(\hat{\phi}_n ) ] =  \int \exp [u_s(\hat{\phi}_n ) \,|\, \overline{u}_n \,] d\overline{F}_n = \int_{A_n^-  (\epsilon )} \exp [u_s(\hat{\phi}_n ) \,|\, \overline{u}_n \,] d\overline{F}_n + \int_{A_n^+ (\epsilon )} \exp [u_s(\hat{\phi}_n) \,|\, \overline{u}_n \,] d\overline{F}_n.
\end{align*}
Since seats are treated symmetrically in a random serial dictatorship, it follows that $ \exp [u_s(\hat{\phi}_n) \,|\, \overline{u}_n \,] = \overline{u}_n$. So for every seat $s$,
\begin{align*}
	\exp[u_s(\hat{\phi}_n ) ] & = \int_{A_n^- (\epsilon )} \overline{u}_n \; d\overline{F}_n + \int_{A_n^+ (\epsilon )} \overline{u}_n \; d\overline{F}_n \\
	& \ge P(A_n^- (\epsilon )) \cdot 0 \, + \, P(A_n^+ (\epsilon )) \bigg [\sum_{i=1}^{I-1} \alpha_iV(C_i,1) + (1- \sum_{i=1}^{I-1} \alpha_i)V(C_I,1) - \frac{1}{2} \epsilon \bigg].
\end{align*}
By \cref{lemma: CT}, $P(A_n^+(\epsilon )) \rightarrow 1$ as $n \rightarrow \infty$. As the result, there exists $N$ such that $\exp[u_s(\hat{\phi}_n ) ] \ge  \sum_{i=1}^{I-1} \alpha_iV(C_i,1) + (1- \sum_{i=1}^{I-1} \alpha_i)V(C_I,1) - \epsilon$ for all $n>N$ and $s\in \hospitals^k_n$.
\end{proof}


\subsubsection{Punishment for Deviation} \label{section: punishment}

\begin{definition} \label{definition:minmax matching}
	Given a set of hospitals $\hospitals'$, students $\students'$, and a hospital $f \in \hospitals'$, the punitive matching for $f$, $\underline{\phi}_n^f$, is the matching produced by the following procedure
\begin{enumerate}
	\item Set $S_0 = \emptyset$ and $G_0 = \emptyset$.
	\item If either $\hospitals' \backslash S_{k-1} = \emptyset $ or $\students' \backslash G_{k-1} = \emptyset$, stop; otherwise, let $\hat{w}$ be $f$'s favorite student in $\students' \backslash G_{k-1}$, and $\hat{f}$ be $\hat{w}$'s favorite hospital among  $\hospitals' \backslash S_{k-1}$. Match $\hat{f}$ and $\hat{w}$, set $ S_{k} =  S_{k-1} \cup \{ \hat{f} \}$ and $G_{k} = G_{k-1}\cup \{\hat{w} \} $. Go to Step $k+1$.
\end{enumerate}
\end{definition}

For a seat $s$ belonging to a hospital $f$, let $\underline R^{s}_n$ denote $f$'s ranking of the student matched to seat $s$ in the matching $\underline{\phi}^f_n$: that is, $\underline R^{s}_n = j$ if $\underline{\phi}^f_n(s)$ is the $j$-th favorite student according to the preference of $f$. The next lemma shows that $R^{s}_n$ is uniformly distributed.

\begin{lemma} \label{lemma: seat distribution}
	For every seat $s$ belonging to $f$ and $ 1\le r \le \abs{\hospitals^k_n}q$, $P(\underline R_n^s = r) = \frac{1}{\abs{\hospitals^k_n}q}$.
\end{lemma}
\begin{proof}
In the procedure in \cref{definition:minmax matching}, starting from $f$'s favorite student, each student takes turn to pick a seat from the remaining seats. As we move down $f$'s preference list, the first student to pick  seat $s$ will determine $\underline R^{s}_n$. 

We can equivalently think of $\underline R^{s}_n$ as being determined through sampling \textit{without} replacement from an urn that contains a red ball (representing seat $s$) and $(\abs{\hospitals^k_n}q - 1)$ black balls (representing all other seats). Specifically, $\underline R^{s}_n$ is the number of draws it takes for the red ball to be sampled. Since the red ball is equally likely to be in any position in the sequence of balls to be drawn, $P(\underline R_n^s = r) = \frac{1}{\abs{\hospitals^k_n}q}$ for $ 1\le r \le \abs{\hospitals^k_n}q$.
\end{proof}

\cref{lemma: clustering} shows that if a hospital is excluded from the very top section of its preference list, then with high probability it obtains a low utility.

\begin{lemma} \label{lemma: clustering}
	Fix any $0< \epsilon <1 $ and $0<\gamma <\epsilon$. For every $f\in \hospitals^k_n$, let $X^f_n$ denote the event that, according to $f$'s preference, the worst $\gamma \abs{\mathcal{V}^I_n}$ students in $\mathcal{V}^I_n$ all satisfy $\zeta_{f,w} < \epsilon$. Then $P(X^f_n) \rightarrow 1 $ as $n\rightarrow \infty$.
\end{lemma}
\begin{proof}
	For each $f\in \hospitals^k_n$ let $K^f_n$ be the number of students in $\mathcal{V}^I_n$ such that $\zeta_{f,w} < \epsilon$. I first prove that $P(K^f_n < \gamma \abs{\mathcal{V}^I_n} ) \rightarrow 0$ as $n\rightarrow \infty$.

	Since $\zeta_{f,w}$ is uniformly distributed over $[0,1]$, $K^f_n$ follows  binomial distribution $B(\abs{\mathcal{V}^I_n}, \epsilon )$. By Hoeffding's inequality, for every $f\in \hospitals^k_n$,
	\begin{align*}
		P(K^f_n < \gamma \abs{\mathcal{V}^I_n} ) \le \frac{1}{2} \text{exp} \Big \{{-2\frac{(\epsilon\abs{\mathcal{V}^I_n} - \gamma\abs{\mathcal{V}^I_n})^2}{\abs{\mathcal{V}^I_n}}} \Big \} \rightarrow 0
	\end{align*}
as $n \rightarrow \infty$. The claim of the lemma follows since $\{K^f_n \ge \gamma \abs{\mathcal{V}^I_n} \} \subseteq X^f_n $.
\end{proof}

\cref{lemma: punishment effective} proves that as market gets large, the payoff $f$ obtains from the punishment algorithm in \cref{definition:minmax matching} is bounded away from what it obtains from the reward matching $\hat{\phi}_n$.

\begin{lemma} \label{lemma: punishment effective}
There exists $g >0$ and $N$ such that 
\begin{equation*}
\exp [u_f(\underline{\phi}_n^f )] < q \big[ \sum_{i=1}^{I-1} \alpha_iV(C_i,1) + (1- \sum_{i=1}^{I-1} \alpha_i) V(C_I,1) \big]	 - g
\end{equation*}
for all $f\in \hospitals^k_n$ and $n>N$.
\end{lemma}
\begin{proof}
To prove the claim, it suffices to prove that there exists $g>0$ and $N$ such that if $n>N$, then
\begin{equation*}
\exp [u_s(\underline{\phi}^f_n )] < \sum_{i=1}^{I-1} \alpha_iV(C_i,1) + (1- \sum_{i=1}^{I-1} \alpha_i)V(C_I,1) - g
\end{equation*}
for every hospital $f\in \hospitals^k_n$ and seat $s$ belonging to $f$. This is what I will prove below.

For every $f\in \hospitals^k_n$, every seat $s$ belonging to $f$, and $1 \le i \le I$, let 
\begin{equation*}
A^s_n(i) \equiv \big\{ \sum_{j\le i-1} \abs{\mathcal{V}^j_n}  < \underline R^s_n \le \sum_{j\le i} \abs{\mathcal{V}^j_n} \big \}	
\end{equation*}
denote the event that $s$ is matched to a student in $\mathcal{V}^i_n$ in the matching $\underline{\phi}^f_n$. By \cref{lemma: seat distribution}, $\underline R^s_n $ is uniformly distributed for every $s$ and $f$, so  $P(A^s_n(i))$ is independent of $f$ or $s$, and I will write $p^i_n \equiv P(A^s_n(i))$ for every $1 \le i \le I$.

Fix an $ { (\sum_{i=1}^{I} \alpha_i -1)}/{\alpha_I} < \epsilon <1$ and ${ (\sum_{i=1}^{I} \alpha_i -1)}/{\alpha_I} < \gamma < \epsilon$.\footnote{Note that $\alpha_I >0$ , and $0 < { (\sum_{i=1}^{I} \alpha_i -1)}/{\alpha_I} = 1 - {(1- \sum_{i=1}^{I-1} \alpha_i)}/{\alpha_I} <1$. So such $\epsilon$ always exists.} Since the function $V(C_I,.)$ is strictly increasing, there exists $g_I > 0$ such that $V(C_I, \zeta ) <V(C_I, 1) - g_I$ for all $\zeta < \epsilon$. For every seat $s$ belonging to any $f\in \hospitals^k_n$, let
\begin{equation*}
	\Gamma^s_n \equiv \big \{ \sum_{i=1}^{ I-1} \abs{\mathcal{V}^i_n} + (1 - \gamma) \abs{\mathcal{V}^I_n } \le \underline R^s_n  \le \sum_{i=1}^{I} \abs{\mathcal{V}^i_n} \big \}
\end{equation*}
denote the event that $s$ is matched to the $\gamma$-tail section of $\mathcal{V}^I_n$ in the matching $\underline{\phi}^f_n$. Again, by \cref{lemma: seat distribution}, $\underline R^s_n $ is uniformly distributed so $P( \Gamma^s_n)$ is independent of $f$ or $s$. I will write $p^{\Gamma}_n \equiv P( \Gamma^s_n)$.

For each hospital $f\in \hospitals^k_n$, let $X^f_n$ denote the event that, according to any hospital $f$'s preference, students in the $\gamma$-tail section of $\mathcal{V}^I_n$ all satisfy $V(C_I, \zeta_{f,w} ) <V(C_I, 1) - g_I$. Since all hospitals in $\hospitals^k_n$ are symmetric, the probability $P(X^f_n)$ does not depend on $f$, and I will simply write $p^X_n \equiv P(X^f_n)$.

By the definition of the events $\Gamma^s_n$ and $X^f_n$, we have, for every $f\in \hospitals^k_n$ and seat $s$ belonging to $f$,
\begin{align} \label{inequality: gamma nonzero}
	P \big (u_s( \underline{\phi}^f_n )  < V(C_I, 1) - g_I \, |\, A^s_n(I) \big ) & \ge  P( \Gamma^s_n \cap X^f_n \, |\, A^s_n(I) )\nonumber \\
	& = P( \Gamma^s_n \cap X^f_n \cap A^s_n(I) ) \, / \, P(A^s_n(I)) \nonumber \\
	& = P( \Gamma^s_n \cap X^f_n  ) \, / \, P(A^s_n(I)) \nonumber \\
	& = \Big[ p^{\Gamma}_n  - P \big(\Gamma^s_n \cap (X^f_n)^c \big) \Big ]\,/ \, p^{I}_n \nonumber \\
	& \ge \Big[ p^{\Gamma}_n - (1 - p^X_n) \big) \Big ]\,/ \, p^{I}_n
\end{align}
Note that the second expression above follows since $\Gamma^s_n \subseteq A^s_n(I)$.

As $n\rightarrow \infty$, $p^I_n \rightarrow 1 - \sum_{i=1}^{I-1} \alpha_i > 0$ and $p^{\Gamma}_n \rightarrow 1- \sum_{i=1}^{I-1} \alpha_i  - (1-\gamma)\alpha_I >0$. Importantly, these limits are all strictly positive numbers.\footnote{To see why, note that by our choice of $\gamma$, $p^{\Gamma}_n = 1- \sum_{i=1}^{I-1} \alpha_i  - (1-\gamma)\alpha_I > 1- \sum_{i=1}^{I-1} \alpha_i - [1 - { (\sum_{i=1}^{I} \alpha_i -1)}/{\alpha_I}] \alpha_I  = 0$} In addition, by \cref{lemma: clustering}, $p^X_n \rightarrow 1$ as $n \rightarrow \infty$. So inequality \eqref{inequality: gamma nonzero} implies that there exits $N_I$such that if $n \ge N_I$,
\begin{equation*}
P \Big( u_s( \underline{\phi}^f_n )  < V(C_I, 1) - g_I \;| \; A^s_n(I) \Big) \ge \frac{p^{\Gamma}_n }{2 p^I_n},
\end{equation*}
for all $f\in \hospitals^k_n$ and seat $s$ belonging to $f$, and therefore
\begin{equation} \label{inequality: gamma expectation}
\exp \big( u_s( \underline{\phi}^f_n )  \;| \; A^s_n(I) \big) < V(C_I, 1) - g_I \frac{p^{\Gamma}_n }{2 p^I_n}.
\end{equation}
Meanwhile, for $1 \le i \le I-1$, $p^i_n \rightarrow \alpha_i \ge 0$ and
\begin{equation} \label{inequality: non gamma expectation}
\exp \big( u_s( \underline{\phi}^f_n )  \;| \; A^s_n(i) \big) \le  V(C_i, 1).
\end{equation}

Define $g \equiv  (1- \sum_{i=1}^{I-1}  \alpha_i) g_I \frac{p^{\Gamma}_n }{2 p^I_n} >0$. Combining inequalities \eqref{inequality: gamma expectation} and \eqref{inequality: non gamma expectation}, there must exist $N$ such that for all $n > N$,
\begin{align*}
\exp \big( u_s( \underline{\phi}^f_n ) \big) & = \sum_{1\le i\le m} p^i_n \;\exp \big( u_s( \underline{\phi}^f_n )  \;| \; A^s_n(i) \big) \\
& < \sum_{i=1}^{I-1} \alpha_iV(C_i,1) + (1- \sum_{i=1}^{I-1} \alpha_i)V(C_I,1) - g
\end{align*}
for all $f\in \hospitals^k_n$ and seat $s$ belonging to $f$. This completes the proof.
\end{proof}

Finally, \cref{lemma: no deviation while punished} proves that when a hospital $f$ is being punished, it cannot find any profitable deviations with students.

\begin{lemma} \label{lemma: no deviation while punished}
For every $f\in \hospitals^k_n$,  no coalition in the form of $(f, W)$ can be a profitable deviation from the matching $\underline{\phi}^f_n$.
\end{lemma}
\begin{proof}
Suppose there is a profitable coalition $(f, W)$ for some $W$. Let $w\in W\backslash  \underline{\phi}^f_n(f)$ be any new student in $W$ who is not originally matched to $f$ in $  \underline{\phi}^f_n$. I will show that $u_f(w') > u_f(w) $ for all $w'\in \underline{\phi}^f_n(f)$. This will be a contradiction to $(f, W)$ being a profitable coalition, because in this case any new student in $W$ is worse than the student in $\underline{\phi}^f_n(f)$ she replaced. It is then impossible for $f$ to prefer $W$ over $\underline{\phi}^f_n(f)$ since hospitals' preferences are responsive.

To this end, first observe that 
if 
there exists $w'\in \underline{\phi}^f_n(f)$ such that $u_f(w') < u_f(w) $, then it means  that in Step \ref{alg: m hat} in the punishment algorithm in \cref{definition:minmax matching}, $w$ chose $\underline{\phi}^f_n(w)$ over $f$ when $f$ still had vacancy (which would later be filled by $w'$) so $\underline{\phi}^f_n(w) \succ_w f$, and $w$ does not find this deviation profitable. This is a contradiction. So  $u_f(w') > u_f(w) $ for all $w'\in \underline{\phi}^f_n(f)$. This completes the proof.
\end{proof}

\section{Supplementary Appendix (Not for Publication)}
\subsection{Preliminary Lemmas} \label{section: supp prelim}

\begin{proofof}{\cref{lemma: one-shot-deviation}}
At each history $h$, hospital $f$ faces a decision over whether and how to deviate with students that are currently in the market. Suppose hospital $f$ has a deviation plan $d_f$ from matching process $\mu$ that is both feasible and profitable, since stage-game payoffs are bounded for hospital $f$ and there is discounting, the standard one-shot deviation principle for individual decision making \citep{blackwell1965discounted} implies that there exists a history $\hat{h} \in \histories$ such that
\begin{equation*}
	(1-\delta)\tilde{u}_f(d_f(\hat{h}) ) + \delta U_f\big(\hat{h}, \big[\mu(\hat{h}),(f, d_f(\hat{h}) ) \big]\, \big | \, \mu \big) > U_f(\hat{h}| \mu ).
\end{equation*}
Consider the deviation plan $d^o_f$ that satisfies $d^o_f(h) = d_f(h)$ if $h =\hat{h}$, and $d^o_f(h) = \mu(f|h) $ otherwise. $d^o_f$ is a profitable one-shot deviation plan for hospital $f$.
\end{proofof}

\cref{lemma: identifiability of manipulator} shows that if a hospital deviates from the recommended matching with a group of students, then this deviating hospital can be uniquely identified by comparing the resulting matching against the recommended matching.

\begin{lemma}\label{lemma: identifiability of manipulator}
	Let $m$ be a static matching. 
If $[m, (f, W)] = [m, (f', W')] \ne m $, then $f = f'$.
\end{lemma}

\begin{proof}
	Let $\overline{m}\equiv [m, (f, W)] $, and $\hat{m} \equiv [m, (f', W')]$. 
Suppose by contradiction that $f \ne f'$, but $\hat{m} = \overline{m}  \ne m $. There are three cases to consider: 
\begin{enumerate}
	\item If $W \subseteq m (f )$: then $\overline{m}(f')  =  m( f' ) \ne  \hat{m}(f') $, so $\overline{m} \ne \hat{m}$, a contradiction.
	\item If $W \nsubseteq m(f)$: then $\overline{m}(f) \nsubseteq m(f) $ but $\hat{m}(f) \subseteq m(f)$, so $\overline{m} \ne \hat{m}$, a contradiction. 
\end{enumerate}
Therefore, $[m, (f, W )] = [m, (f', W' )] \ne m $ implies $f = f'$.
\end{proof}

\subsection{Proof of \cref{theorem:folk-theorem}}
Fix $\lambda^0 \in \Lambda^* $. Define $u^0 = u({\lambda^0}) $ and $U^* \equiv \{u(\lambda):\lambda \in \Lambda^* \}$. Observe that for hospitals in $\hospitals\cap\reduced$, the set $\redmatching^{\circ}$ satisfies the non-equivalent utilities (NEU) condition in \citet{abreuduttasmith}: holding $f\in \hospitals\cap \reduced$ unmatched, $f$ is indifferent towards how another hospital $f'\in \hospitals\cap \reduced$ matches with students, so their utilities cannot be positive affine transformation of another. Lemma 1 and Lemma 2 in \citet{abreuduttasmith} then ensure the existence of vectors $\{ u^{f}: f\in \hospitals\cap\reduced  \} \subseteq U^*$, such that
\begin{equation*}
u^{f}_f < u^0_f \,\, \text{ and }\,\, u^{f}_f < u^{f'}_f
\end{equation*}
for all $f, f' \in \hospitals\cap\reduced$ and $f \ne f'$. Let $\lambda^f \in \Lambda^*$ be the distribution over $\redmatching$ that give rise to the payoff vector $u^f$ for each $f$. In addition, for each $f\in \hospitals \cap \reduced$, let 
\begin{equation*}
\underline{m}_f  \equiv \argmin_{m \in \redmatching^{\circ}} \max_{ W \subseteq D_f(m), |W| \le q_f } u_f(W)	
\end{equation*}
be the stage-game recommendation to minmax hospital $f$.

Consider the matching process represented by the automaton $(\Theta, \gamma^0, f, \gamma)$, where
\begin{enumerate}
	\item $\Theta = \big \{\theta(e,m): e \in \hospitals\cap\reduced \cup \{ 0 \} , m  \in \redmatching \big \} \cup \big \{ \underline{\theta}(f, t): f \in \hospitals\cap\reduced , \; 0\le t < L \big \}$ is the set of all states;
	\item $\gamma^0$ is the initial distribution over states, which satisfies $\gamma^0( \theta(0,m) ) = \lambda^0(m ) $ for all $m \in \redmatching^{\circ}$;
	\item $O:\Theta \rightarrow \smatchings $ is the output function, where $O(\theta(e, m )) = m $ and $O(\underline{\theta}(f, t))= \underline{m}_f$;
	\item $\kappa: \Theta \times \smatchings \rightarrow \Delta(\Theta)$ is the transition function.
	For states $\{\underline{\theta}(f, t)| 0 \le t < L -1 \}$, $\kappa$ is defined as
	\[
	\kappa \big( \underline{\theta} (f, t ), m' \big) =
	\begin{cases}
	\underline{\theta}(f', 0) & \text{if $m' \ne \underline{m}_f$; $m' = [\underline{m}_{f'}, (f', W) ]$ for some  $f' \in \hospitals\cap\reduced$ and $W \subseteq \students$}\\
	\underline{\theta}(f, t +1) & \text{otherwise}
	\end{cases}
	\]
	For states $\underline{\theta}(f, L -1)$, the transition is defined as
	\[
	\kappa \big( \underline{\theta}(f, L -1), m' \big) =
	\begin{cases}
	\underline{\theta} (f', 0) & \text{if $m' \ne \underline{m}_f$; $m' = [\underline{m}_{f'}, (f', W) ]$ for some  $f' \in \hospitals\cap\reduced$ and $W \subseteq \students$} \\
		\gamma^f & \text{otherwise}
	\end{cases}
	\]
	where for each $f\in \hospitals\cap\reduced$, $p^f$ is the distribution over states that satisfies $\gamma^f(\theta(f, m) ) = \lambda^f (m)$ for all $m \in \smatchings$.
	
	For states $\theta (e,m )$, the transition is
	\[
	\kappa \big( \theta (e, m ), m' \big) =
	\begin{cases}
	\underline{\theta}(f',0) & \text{if $m' \ne \underline{m}_f$; $m' = [\underline{m}_{f'}, (f', W) ]$ for some  $f' \in \hospitals\cap\reduced$ and $W \subseteq \students$} \\
	\gamma^e & \text{otherwise}
	\end{cases}
	\]			
	where the distributions $\gamma^e$ are defined as above.
\end{enumerate}
Note that owing to the identifiability of deviating hospital (\cref{lemma: identifiability of manipulator}), for any $\theta \in \Theta$ and matching $m' \ne O(\theta)$, we can uniquely identify the hospital responsible for this deviation, so the transition above is well-defined. 


Note that no hospitals in $\tcs$ wish to deviate, since they are always matched with their top coalition students; no students want to deviate since all recommended matchings are in $\redmatching^{\circ}$. It remains to verify no hospital $f\in \hospitals\cap\reduced$ has incentives to deviate. Choose a number $Z >\sup_{ \{ m \in \smatchings , f\in \hospitals\cap\reduced \} } u_{f}(m ) $
\bigskip
	
\noindent \textbf {{For states of the form $\big \{\theta(e,m) \big \}$}:} Consider a one-shot deviation $(f, W)$. There are two cases to consider.

\smallskip
\noindent \textit{Case 1:} $f \ne e$. 
Without deviation, $f$ has value $(1-\delta) {u}_{f}(m) + \delta  u_{f}^{e} $. After deviation, $f$ yields less than $(1-\delta) Z +\delta(1-\delta^{L }) \underline{u}^{\reduced}_{f} +\delta^{L +1} u^{f}_{f} $. There is no profitable one-shot deviation for $f$ if
	\[
	(1-\delta) u_{f}(m) + \delta u_{f}^{e} \ge (1-\delta) Z +\delta(1-\delta^{L }) \underline{u}^{\reduced}_{f} +\delta^{L +1} u^{f}_{f} 
	\]
	As $\delta \rightarrow 1$, the LHS converges to $u^e_{f}$ while the RHS converges to $u^{f}_{f}$. By construction, $u^e_{f}> u^{f}_{f}$, so such deviations are not profitable for $\delta$ high enough.
\medskip
	
\noindent \textit{Case 2:}  $f = e$.  Without deviation, $f$ has value $(1-\delta) u_{f}(m) + \delta u_{f}^{f} $. After deviation, $f$ yields less than $(1-\delta) Z +\delta(1-\delta^{L }) \underline{u}^{\reduced}_{f} +\delta^{L +1} u^{f}_{f}$.
There is no profitable one-shot deviation for $f$ if
\[
	(1-\delta) u_{f}(m) + \delta u_{f}^{f} \ge (1-\delta) Z +\delta(1-\delta^{L }) \underline{u}^{\reduced}_{f} +\delta^{L +1} u^{f}_{f}.
\]
The inequality is equivalent to
\[
 Z - u_{f}(m )   \le \delta(1 + \ldots + \delta^{L-1} ) [u^{f}_{f} - \underline{u}^{\reduced}_{f} ]
\]
By construction, $u^{f}_{f} - \underline{u}^{\reduced}_{f}>0$. Choose $L$ large enough so that $L (u^{f}_{f} - \underline{u}^{\reduced}_{f}) >  Z - u_{f}(m )$. As $\delta \rightarrow 1$, the LHS remains unchanged while the RHS converges to $L (u^{f}_{f} - \underline{u}^{\reduced}_{f})$, so such deviations are not profitable for $\delta$ high enough.
\bigskip

\noindent \textbf{For states of the form $\big \{ \underline{\theta}(f, t) \big \}$:} Consider a one-shot deviation $(f', W)$. There are two cases to consider.
\smallskip

	\noindent \textit{Case 1:} $f' \ne f$. Without deviation, hospital $f'$ has payoff $(1-\delta^{L - t} ) u_{f'} (\underline{m}_f)  + \delta^{L - t} u_{f'}^{f}$. After deviation, $f'$ has payoff less than $(1-\delta) Z + \delta(1-\delta^{L}) \underline{u}^{\reduced}_{f'}  +\delta^{L +1} u_{f'}^{f'}$. There is no profitable one-shot deviation for $f'$ if
	\begin{equation*}
	(1-\delta^{L - t} ) u_{f'} (\underline{m}_f)  + \delta^{L - t} u_{f'}^{f} \ge (1-\delta) Z + \delta(1-\delta^{L}) \underline{u}^{\reduced}_{f'}  +\delta^{L +1} u_{f'}^{f'}
	\end{equation*}
	As $\delta \rightarrow 1$, the LHS converges to $u^{f}_{f'}$ for all $0\le t \le L$, while the RHS converges to $u^{f'}_{f'}$. By construction $u^{f}_{f'} > u^{f'}_{f'}$. So the above inequality holds for sufficiently high $\delta$. 
\medskip

\noindent \textit{Case 2:} $f' = f$. Without deviation, hospital $f'$ has payoff $(1-\delta^{L - t}) \underline{u}^{\reduced}_{f'} + \delta^{L - t} u_{f'}^{f'}$. When deviating from $\underline{m}_{f'}$, $f'$ can obtain at most $\underline{u}^{\reduced}_f$. So its payoff from deviation is at most
\[
(1-\delta) \underline{u}^{\reduced}_{f'} + \delta(1-\delta^{L}) \underline{u}^{\reduced}_{f'} +\delta^{L +1} u_{f'}^{f'} = (1 - \delta^{L+1} ) \underline{u}^{\reduced}_{f'} + \delta^{L +1} u_{f'}^{f'} 
\]
Hospital $f'$ has no profitable deviation if
$ (1-\delta^{L - t}) \underline{u}^{\reduced}_{f'}  + \delta^{L - t} u_{f'}^{f'} \ge (1 - \delta^{L+1} ) \underline{u}^{\reduced}_{f'} + \delta^{L +1} u_{f'}^{f'}$, or
\[
	u_{f'}^{f'} \ge \underline{u}^{\reduced}_{f'}.
\]
This is true by construction. So $f'$ has no profitable one-shot deviation.

We have verified that there is no profitable one-shot deviation in any states of the automaton. This completes the proof.

\subsection{Completion of the Proof of \cref{theorem: can reduce capacity}} \label{section: supp matching process}

Recall that $Q^{\students}_n (l) \equiv \sum_{l'\le l} |\students^{l'}_n |$ denote the number of students who are in a quality class no worse than $\students^l_n $, while $Q^{\hospitals}_n (k) \equiv q \sum_{k'\le k } |\hospitals^{k'}_n |$ denote the number of hospital seats that are in a quality class no worse than $\hospitals^k_n$. Let $\mathcal{A}(k) = \{ j+1, \ldots, j+I \}$ be the set of student quality classes achievable by hospitals in quality class $k$. By definition, there exists $N_0$ such that for all $n>N_0$ and every $l \notin \mathcal{A}(k)$,

\begin{equation} \label{equation: non-achievable out}
\text{either } \;\;  Q^{\students}_n(l) < Q^{\hospitals}_n(k-1)  \;\; \text{ or }   \;\; Q^{\students}_n(l-1) > Q^{\hospitals}_n(k).
\end{equation} 
I focus on market sizes greater than this $N_0$.
By the inequalities in \eqref{equation: non-achievable out}, in all the matchings defined below, students in any unachievable quality class by $\hospitals^k_n$ are either already matched to hospitals in a higher quality class, or ranked too low to be relevant for $\hospitals^k_n$. Therefore, to analyze the payoff for a hospital $f\in \hospitals^k_n$, it suffices to isolate $\hospitals^k_n$ and $\cup_{l\in \mathcal{A}(k)} \students^l_n$ as if they are the only hospitals and students in the market. For each $l \in \mathcal{A}(k)$, let
\begin{equation*}
\alpha_{l} \equiv \lim_{n \rightarrow \infty} \Bigg\{ \frac{ Q^{\students}_n(l) - \max \{ Q^{\students}_n(l-1), Q^{\hospitals}_n(k-1) \} } {q|\hospitals^k_n |} \Bigg\}.
\end{equation*}
denote the asymptotic share of seats in $\hospitals^k_n$ filled by students in $\students^l_n$. Since $\abs{\students_n} = \lceil \beta nq \rceil$, we know all $\alpha_l$'s are finite numbers. Note that they also satisfy $ \sum_{i=1}^{I-1}\alpha_{j+i} <1 $ and $ \sum_{i=1}^{I}\alpha_{j+i} \ge 1 $.



\paragraph{Minmax Matching.} For each $f\in \hospitals^k_n$, consider the matching $\underline{m}_n^{f}$ defined by the following procedure:
\begin{enumerate}
	\item Let ${\phi}_n^{<}$ be a stable matching between $\cup_{j<k}\hospitals^j_n$ and $\students_n$. Set $\underline{m}_n^{f}(f') = {\phi}^{<}(f')$ for every $f' \in \cup_{j<k}\hospitals^j_n$;
	\item Let $\underline{\phi}_n^f$ be the punitive matching for $f$ among $\hospitals^k_n$ and $\students_n\backslash\, {\phi}_n^{<}(\cup_{j<k} \hospitals^j_n )$. Set $\underline{m}^f_n(f') = \underline{\phi}^f_n(f')$ for every $f' \in \hospitals^k_n$; \label{alg: m hat}
	\item Let ${\phi}_n^{>}$ be a stable matching between $\cup_{j>k}\hospitals^j_n$ and $\students_n\backslash \big[ {\phi}_n^{<}(\cup_{j<k} \hospitals^j_n )\cup \underline{\phi}_n^f(\hospitals^k_n ) \big]$. Set $\underline{m}_n^f(f') = {\phi}_n^{>}(f')$ for every $f' \in \cup_{j>k}\hospitals^j_n$.
\end{enumerate}

By \cref{lemma: punishment effective}, there exists $\underline{g}>0$ and market size $N_1$ such that $\exp [u_f(\underline{m}^f_n)] < \sum_{i = 1}^{I-1} \alpha_{j+i} V(C_{j+i}, 1 ) + (1 - \sum_{i = 1}^{I-1} \alpha_{j+i} )V(C_{j+I}, 1) - \underline{g} $ for all $f \in \hospitals^k_n$ and $n>N_1$.



\paragraph{Reward Matching.}
Consider the matching $\hat{m}_n$ defined by the following procedure:
\begin{enumerate}
	\item Let ${\phi}^{<}_n$ be a stable matching between $\cup_{j<k}\hospitals^j_n$ and $\students_n$. Set $\hat{m}_n(f) = {\phi}_n^{<}(f)$ for every $f \in \cup_{j<k}\hospitals^j_n$;
	\item Let $\hat{\phi}_n$ be the matching resulting from hospital-proposing random serial dictatorship between $\hospitals^k_n$ and $\students_n\backslash {\phi}_n^{<}(\cup_{j<k} \hospitals^j_n )$. Set $\hat{m}_n(f) = \hat{\phi}_n (f) $ for all $f \in \hospitals^k_n$; 
	\item Let $\phi^{>}_n$ be a stable matching between $\cup_{j>k}\hospitals^j_n$ and $\students_n\backslash \big[ \phi_n^{<}(\cup_{j<k} \hospitals^j_n )\cup \hat{\phi}_n (\hospitals^k_n ) \big]$. Set $\hat{m}_n(f) = \phi_n^{>}(f)$ for every $f \in \cup_{j>k}\hospitals^j_n$.
\end{enumerate}


By \cref{lemma: rsd high payoffs}, there exists  market size $N_2$ such that $\exp [u_f(\hat{m}_n)] > \sum_{i = 1}^{I-1} \alpha_{j+i} V(C_{j+i}, 1 ) + (1 - \sum_{i = 1}^{I-1} \alpha_{j+i} )V(C_{j+I}, 1) - \frac{1}{2} \underline{g}$ for all $f \in \hospitals^k_n$ and $n>N_2$.

I focus on market sizes greater than $N \equiv \max\{N_0, N_1, N_2\} $. In particular, 
\begin{equation} \label{inequality: gap between reward and punishment}
	\exp [u_f(\hat{m}_n )] > \exp [u_f(\underline{m}^f_n )] + \frac{1}{2} \underline{g}
\end{equation}
for all $n>N$ and all $f\in \hospitals^k_n$.

\paragraph{Matching with Reduced Capacity.} Consider the matching ${m}^*_n$ defined by the following procedure:
\begin{enumerate}
	\item Let ${\phi}^{<}_n$ be a stable matching between $\cup_{j<k}\hospitals^j_n$ and $\students_n$. Set $m^*_n(f) = {\phi}_n^{<}(f)$ for every $f \in \cup_{j<k}\hospitals^j_n$;
	\item Let $\phi^*_n$ be the student-proposing stable matching between $\hospitals^k_n$ and $\students_n\backslash {\phi}_n^{<}(\cup_{j<k} \hospitals^j_n )$, where instead of letting every hospital use up its hiring capacity $q$, I treat each hospital as if only having $(q-1)$ quotas. Set $m^*_n(f) = \phi^*_n(f)$ for every $f\in \hospitals^k_n$.
	\item Let $\phi^{>}_n$ be a stable matching between $\cup_{j>k}\hospitals^j_n$ and $\students_n\backslash \big[ \phi_n^{<}(\cup_{j<k} \hospitals^j_n )\cup \phi^*_n (\hospitals^k_n ) \big]$. Set $m^*_n(f) = \phi_n^{>}(f)$ for every $f \in \cup_{j>k}\hospitals^j_n$.
\end{enumerate}
Observe that $\exp [u_f( m^*_n)] \ge 0 $ for every $f\in \hospitals^k_n$. Combine this with inequality \eqref{inequality: gap between reward and punishment}, we know there exists $p^0 \in (0,1)$ and $\epsilon_0>0$ such that 
\begin{equation*} 
	p^0 \exp [u_f(m^*_n)] + (1- p^0) \exp [u_f(\hat{m}_n )] > \exp [u_f(\underline{m}^f_n)] + \epsilon_0
\end{equation*}
for all $n>N$ and every $f\in \hospitals^k_n$. Let $\lambda^0_n$ denote the lottery over $m^*_n$ and $\hat{m}_n$ with $	p^0$ weight on $m^*_n$ and $(1-p^0)$ weight on $\hat{m}_n$, then the inequality above is equivalent to
\begin{equation} \label{inequality: target lottery non vanishing gap}
	\exp [u_f(\lambda^0_n)]  > \exp [u_f(\underline{m}^f_n)] + \epsilon_0	
\end{equation}
for all $n>N$ and every $f\in \hospitals^k_n$.

\paragraph{Hospital-Specific Punishments.}
For each hospital $f\in \hospitals^k_n$, consider the matching ${m}^f_n$ defined by the following procedure:
\begin{enumerate}
	\item Let ${\phi}^{<}_n$ be a stable matching between $\cup_{j<k}\hospitals^j_n$ and $\students_n$. Set $m^f_n(f') = {\phi}_n^{<}(f')$ for every $f' \in \cup_{j<k}\hospitals^j_n$;
	\item Let $\phi^f_n$ be the student-proposing stable matching between $\hospitals^k_n$ and $\students_n\backslash {\phi}_n^{<}(\cup_{j<k} \hospitals^j_n )$, where we treat hospital $f$ as if having $0$ quota, while others are allowed to hire up to full capacity. Set $m^f_n(f') = \phi^f_n(f')$ for every $f'\in \hospitals^k_n$.
	\item Let $\phi^{>}_n$ be a stable matching between $\cup_{j>k}\hospitals^j_n$ and $\students_n\backslash \big[ \phi_n^{<}(\cup_{j<k} \hospitals^j_n )\cup \phi^f_n (\hospitals^k_n ) \big]$. Set $m^f_n(f') = \phi_n^{>}(f')$ for every $f' \in \cup_{j>k}\hospitals^j_n$.
\end{enumerate}
Observe that $\exp [u_{f'}(m^f_n)]\ge 0$ for all $n> N$ and $f, f'\in \hospitals^k_n$. Combine this with inequality \eqref{inequality: target lottery non vanishing gap}, we know there exists $p^r \in (0,1)$ and $\epsilon_1$ such that for all $n>N$ and $f,f' \in \hospitals^k_n$,
\begin{equation*}
	p^r \exp [u_f( \lambda^0_n ) ] + (1 -p^r ) \exp [ u_f(m^{f'}_n )] > \exp [u_f(\underline{m}^f_n)] + \epsilon_1 
\end{equation*}
Also since $\exp [u_f( m^f_n)] = 0 $ for all $n>N$ and $f\in \hospitals^k_n$, combined with inequality \eqref{inequality: target lottery non vanishing gap}, we know there exists $\epsilon_2$ such that for all $f \in \hospitals^k_n$,
\begin{equation*}
	\exp [u_f(\lambda^0_n )] > p^r \exp [u_f( \lambda^0_n ) ] + (1 -p^r ) \exp [ u_f(m^f_n )] + \epsilon_2   
\end{equation*}
Lastly, since $\exp [u_{f'}(m^f_n)] \ge qV(C_1, 0) >  0 = \exp [u_f( m^f_n)]$ for all $n>N$ and $f,f'\in \hospitals^k_n$ where $f \ne f'$, combined with inequality \eqref{inequality: target lottery non vanishing gap}, we know there exists $\epsilon_3$ such that for all $n>N$ and all $f,f' \in \hospitals^k_n$ with $f \ne f'$,
\begin{equation*}
 p^r \exp [u_f( \lambda^0_n ) ] + (1 -p^r ) \exp [ u_f(m^{f'}_n )] 	 > p^r \exp [u_f( \lambda^0_n ) ] + (1 -p^r ) \exp [ u_f(m^f_n )] + \epsilon_3 
\end{equation*}
Let $\epsilon = \min\{\epsilon_1, \epsilon_2, \epsilon_3 \}$ and let $\lambda^f_n$ denote the compound lottery placing $p^r$ weight on the lottery $\lambda^0_n$ and $(1- p^r)$ on the matching $m^f_n$, the above three inequalities imply that for all $n > N$
\begin{align}
		 \exp [u_f( \lambda^{f'}_n ) ] &> \exp [u_f(\underline{m}^f_n)] + \epsilon \;\; \text{ for all } f, f'\in \hospitals^k_n \label{inequality: p-s punishment1} \\
		 \exp [u_f(\lambda^{f}_n )]  &< \exp [u_f(\lambda^0_n )] - \epsilon \;\; \text{ for all } f\in \hospitals^k_n \label{inequality: p-s punishment2} \\
		 \exp [u_f(\lambda^{f}_n )] &< \exp [u_f(\lambda^{f'}_n )]  - \epsilon \;\; \text{ for all } f, f'\in \hospitals^k_n \text{ and } f \ne f' \label{inequality: p-s punishment3}
\end{align}
Taken together, inequalities \eqref{inequality: p-s punishment1}, \eqref{inequality: p-s punishment2} and \eqref{inequality: p-s punishment3} imply that the matchings $\{\lambda^{f}_n: f\in \hospitals^k_n \}$ form hospital-specific punishments for $\lambda^{0}_n$ for all $n> N$.

\paragraph{A Capacity-Reducing, Self-Enforcing Matching Process.} Consider the matching process represented by the automaton $(\Theta, \gamma^0 , O, \kappa)$, where
\begin{enumerate}
	\item $\Theta_n = \big \{\theta( e ,m): f \in \hospitals^k_n \cup \{ 0 \};\; m  \in {M}_n \big \} \cup \big \{ \underline{\theta}(f, t): f \in \hospitals^k_n , \; 0\le t < L \big \}$ is the set of all possible states;
	\item $\gamma^0$ is the initial distribution over states, which satisfies $\gamma^0 ( \theta(0,m) ) = \lambda^0(m ) $ for all $m \in M_n$;
	\item $O:\Theta \rightarrow M_n$ is the output function, where $O(\theta(f, m )) = m $ and $O(\underline{\theta}(f, t))= \underline{m}^f_n$;
	\item $\kappa: \Theta \times M_n \rightarrow \Delta(\Theta)$ is the transition function defined as follows.

	For states $\{\underline{\theta}(f, t)| 0 \le t < L -1 \}$, $\kappa$ is defined as
	\[
	\kappa \big( \underline{\theta} (f, t ), m' \big) =
	\begin{cases}
	\underline{\theta}(f', 0) & \text{if $m' \ne \underline{m}^f_n$; $m' = [\underline{m}^{f}_n, (f', W) ]$ for some  $f' \in \hospitals^k_n$ and $W \subseteq \students_n $} \\
	\underline{\theta}(f, t +1) & \text{otherwise}
	\end{cases}
	\]
	For states $\underline{\theta}(f, L -1)$, the transition is defined as
	\[
	\kappa \big( \underline{\theta}(f, L -1), m' \big) =
	\begin{cases}
	\underline{\theta} (f', 0) & \text{if $m' \ne \underline{m}^f_n$; $m' = [\underline{m}^{f}_n, (f', W) ]$ for some  $f' \in \hospitals^k_n$ and $W \subseteq \students_n $} \\
		\gamma^f & \text{otherwise}
	\end{cases}
	\]
	where $p^f$ is the distribution over states that satisfies $\gamma ^f(\theta(f, m) ) = \lambda^f (m)$ for all $f \in \hospitals^k_n$ and $m \in M_n$.
		
	For states $\theta(e, m)$, the transition is
	\[
	\kappa \big( \theta (e, m ), m' \big) =
	\begin{cases}
	\underline{\theta}(f',0) & \text{if $m' \ne m$; $m' = [m, (f', W) ]$ for some  $f' \in \hospitals^k_n$ and $W \subseteq \students_n$} \\
		\gamma ^e & \text{otherwise}
	\end{cases}
	\]	
	where $\gamma^f$ and $ \gamma^0$ are defined as above.
\end{enumerate}

\noindent The matching process represented by the above automaton randomizes over $\redmatching$ according to $\lambda^0$ in every period. It remains to  verify no hospital has profitable one shot deviations in any automaton state.

\noindent \textbf {{For states of the form $ \theta(e,m)$}:} Consider a one-shot deviation $(f', W)$ by hospital $f'$. There are two cases to consider.

\noindent \textit{Case 1:} $f'  \ne e$. Choose a number $Z > q V(C_1,1)$, so no hospital can derive payoff higher than $Z$ from any deviation. Without deviation, $f'$ has value $(1-\delta)  u_{f'}(m) + \delta \exp [u_{f'}(\lambda^{e}_n )] $. After deviation, $f'$ yields less than $(1-\delta) Z +\delta(1-\delta^{L }) \exp [u_{f'}(\underline{m}^{f'}_n)] +\delta^{L +1}  \exp [u_{f'}(\lambda^{f'}_n )] $.
There is no profitable one-shot deviation for $f'$ if
\[
(1-\delta)  u_{f'}(m) + \delta \exp [u_{f'}(\lambda^{e}_n )]  \ge (1-\delta) Z +\delta(1-\delta^{L }) \exp [u_{f'}(\underline{m}^{f'}_n)] +\delta^{L +1}  \exp [u_{f'}(\lambda^{f'}_n )]
\]
As $\delta \rightarrow 1$, the LHS converges to $\exp [u_{f'}(\lambda^{e}_n )]$ while the RHS converges to $\exp [u_{f'}(\lambda^{f'}_n )]$. By \eqref{inequality: p-s punishment2} and \eqref{inequality: p-s punishment3}
, $\exp [u_{f'}(\lambda^{e}_n )] > \exp [u_{f'}(\lambda^{f'}_n )] + \epsilon$ for all $n>N$. It follows that there is a $\underline{\delta}_1$ such that these deviations are not profitable for $\delta > \underline{\delta}_1$ and $n>N$.
\medskip

\noindent \textit{Case 2:} $f' = e$.  Without deviation, $f'$ has value $(1-\delta) u_{f'}(m) + \delta \exp [u_{f'}(\lambda^{f'}_n )]$. After deviation, $f'$ yields less than $(1-\delta) Z +\delta(1-\delta^{L }) \exp [u_{f'}(\underline{m}^{f'}_n)] +\delta^{L +1}  \exp [u_{f'}(\lambda^{f'}_n )]$. There is no profitable one-shot deviation for $f'$ if
\[
	(1-\delta) u_{f'}(m) + \delta \exp [u_{f'}(\lambda^{f'}_n )]  \ge (1-\delta) Z +\delta(1-\delta^{L }) \exp [u_{f'}(\underline{m}^{f'}_n)] +\delta^{L +1}  \exp [u_{f'}(\lambda^{f'}_n )].
\]
The inequality is equivalent to
\[
 Z - u_{f'}(m )   \le \delta(1 + \ldots + \delta^{L-1} ) \big[\exp [u_{f'}(\lambda^{f'}_n )] - \exp [u_{f'}(\underline{m}^{f'}_n)]\big].
\]
By \eqref{inequality: p-s punishment1}, $\exp [u_{f'}(\lambda^{f'}_n )] - \exp [u_{f'}(\underline{m}^{f'}_n)] > \epsilon$ for all $n>N$. Choose $L$ large enough so that $L \epsilon >  Z - u_{f'}(m )$. As $\delta \rightarrow 1$, the LHS remains unchanged while the RHS converges to $L (\exp [u_{f'}(\lambda^{f'}_n )] - \exp [u_{f'}(\underline{m}^{f'}_n)]) > L\epsilon$, there exist $\underline{\delta}_2$ such that no deviation is profitable for $\delta > \underline{\delta}_2$ and $n>N$.
\bigskip

\noindent \textbf{For states of the form {$ \underline{\theta}(f, t)$}:} there are two cases to consider.

\noindent \textit{Case 1:} $f' \ne f$. Without deviation, $f'$ has payoff $(1-\delta^{L - t} ) \exp [u_{f'}(\underline{m}^{f}_n)]  + \delta^{L - t} \exp [u_{f'}(\lambda^{f}_n )]$. With any deviation, $f'$ has payoff less than $(1-\delta) Z + \delta(1-\delta^{L})  \exp [u_{f'}(\underline{m}^{f'}_n)]  +\delta^{L +1} \exp [u_{f'}(\lambda^{f'}_n )]$.
There is no profitable one-shot deviation for $f'$ if
\begin{equation*}
(1-\delta^{L - t} ) \exp [u_{f'}(\underline{m}^{f}_n)]  + \delta^{L - t} \exp [u_{f'}(\lambda^{f}_n )] \ge (1-\delta) Z + \delta(1-\delta^{L})  \exp [u_{f'}(\underline{m}^{f'}_n)]  +\delta^{L +1} \exp [u_{f'}(\lambda^{f'}_n )]
\end{equation*}
As $\delta \rightarrow 1$, the LHS converges to $\exp [u_{f'}(\lambda^{f}_n )]$ for all $t$ such that $0\le t \le L$, while the RHS converges to $\exp [u_{f'}(\lambda^{f'}_n )]$. By \eqref{inequality: p-s punishment3}, $\exp [u_{f'}(\lambda^{f}_n )] > \exp [u_{f'}(\lambda^{f'}_n )] +\epsilon$ for all $n>N$. So there exists $\underline{\delta}_3$ such that for all $\delta >\underline{\delta}_3$ and $n>N$ there is no profitable deviations. 
\medskip

\noindent \textit{Case 2:} $f'= f$. Without deviation, hospital $f'$ has payoff
\[
(1-\delta^{L - t}) \exp [u_{f'}(\underline{m}^{f'}_n)]+ \delta^{L - t} \exp [u_{f'}(\lambda^{f'}_n )].
\]
When deviating from $\underline{m}_{k'}$, by \cref{lemma: no deviation while punished}, $f'$'s stage-game payoff is at most $\exp [u_{f'}(\underline{m}^{f'}_n)]$. So $f'$'s discounted expected payoff from deviation is at most
\[
(1-\delta) \exp [u_{f'}(\underline{m}^{f'}_n)] + \delta(1-\delta^{L}) \exp [u_{f'}(\underline{m}^{f'}_n)] +\delta^{L +1} \exp [u_{f'}(\lambda^{f'}_n )] = (1 - \delta^{L+1} ) \exp [u_{f'}(\underline{m}^{f'}_n)]  + \delta^{L +1} \exp [u_{f'}(\lambda^{f'}_n )].
\]
Hospital $f'$ has no profitable deviation if
\begin{equation*} 
(1-\delta^{L - t}) \exp [u_{f'}(\underline{m}^{f'}_n)]+ \delta^{L - t} \exp [u_{f'}(\lambda^{f'}_n )] \ge (1 - \delta^{L+1} ) \exp [u_{f'}(\underline{m}^{f'}_n)]  + \delta^{L +1} \exp [u_{f'}(\lambda^{f'}_n )],
\end{equation*}
or
\[
	\exp [u_{f'}(\lambda^{f'}_n )] > \exp [u_{f'}(\underline{m}^{f'}_n)].
\]
which is true for all $n>N$. So $f'$ has no profitable one-shot deviation for all $n>N$.

Define $\underline{\delta} \equiv \max\{\delta_1, \delta_2, \delta_3 \}$. There is no profitable one-shot deviation in any states of the automaton, as long as $\delta>\underline{\delta}$ and $n>N$. This completes the proof.

\subsection{Can Hospitals Treat Students As Actions?} \label{section: product structure}

In this section I consider a non-cooperative game where hospitals choose students as actions, and discuss its relation to cooperative game form of two-sided matching games. When the matching environment has no transfers and no peer effects, this non-cooperative game yields the same outcome as the cooperative game form of two-sided matching games. However, this equivalence breaks down in matching games with money or externalities. One advantage of the hybrid solution concept in \cref{definition:self-enforcing-matching-process} is therefore its portability to all kinds of matching environments. 

Below I illustrate the difference between matching games and the non-cooperative game form using an example: the key difference is whether or not students are given agency in the formation of blocking coalitions.

\paragraph{A Static Matching Game with Peer Effects.} There are two hospitals $\mathcal{F} = \{f_1, f_2 \}$ and two students $\mathcal{W} = \{w_1, w_2\}$. There is no money. Each hospital has capacity to hire both students. 
\medskip

\noindent \textbf{Hospitals:}
Let $u_k(.)$ be hospital $f_k$'s utility function for students. Assume each hospital derives $1$ util from hiring each student, and both hospitals have additive utilities for the two positions to be filled. Formally,  $u_k(\{w_1\}) = u_k(\{w_2\}) = 1$, and $u_k(\{w_1, w_2\}) = 2$ for $k=1,2$.
\medskip

\noindent \textbf{Students:}
Let $v_j(f, w)$ be student $w_j$'s utility from working for hospital $f$, alongside colleague $w$. Assume $v_j(f_2, \, . \,) = 0$ for $j = 1, 2$, so both students are indifferent between working for $f_2$ and their outside options, regardless of colleagues. Assume $v_j(f_1, \emptyset ) = 2$ and $v_j(f_1, w_i ) = 1$ for all $j=1, 2$ and $i \ne j$: both students enjoy working for $f_1$ but prefer working alone.

\paragraph{Stable Matchings.}
There are three stable matchings, as illustrated in \Cref{figure:cooperative matchings}. Notice that $m_1$ and $m_2$ are stable matchings because the coalition $\{w_1, w_2\}$ is not a feasible deviation for $f_1$: since students have agency in the formation of blocking coalitions, and students dislike colleagues, in both $m_1$ and $m_2$, the coalition $\{f_1, w_1, w_2\}$ is turned down by $f_1$'s existing student.

\begin{figure}[h]  \centering
\minipage{0.2\textwidth} \centering
\begin{tikzpicture}
\tikzstyle{every node} = [draw, shape=circle, line width=1pt, inner sep=0pt, minimum size=20pt]
\tikzset{every path/.append style={line width=1pt}}

\path (0,3) node (f2) {$f_2$};
\path (0,1) node (f1) {$f_1$};
\path (3,3) node (w2) {$w_2$};
\path (3,1) node (w1) {$w_1$};
\draw (f1) -- (w1);
\draw (f2) -- (w2);
\end{tikzpicture}	
\caption*{$m_{1}$}
\endminipage 
\hspace{10ex}
\minipage{0.2\textwidth} \centering
\begin{tikzpicture}
\tikzstyle{every node} = [draw, shape=circle, line width=1pt, inner sep=0pt, minimum size=20pt]
\tikzset{every path/.append style={line width=1pt}}

\path (0,3) node (f2) {$f_2$};
\path (0,1) node (f1) {$f_1$};
\path (3,3) node (w2) {$w_2$};
\path (3,1) node (w1) {$w_1$};
\draw (f1) -- (w2);
\draw (f2) -- (w1);
\end{tikzpicture}	
\caption*{$m_{2}$}
\endminipage 
\hspace{10ex}
\minipage{0.2\textwidth} \centering
\begin{tikzpicture}
\tikzstyle{every node} = [draw, shape=circle, line width=1pt, inner sep=0pt, minimum size=20pt]
\tikzset{every path/.append style={line width=1pt}}
\path (0,3) node (f2) {$f_2$};
\path (0,1) node (f1) {$f_1$};
\path (3,3) node (w2) {$w_2$};
\path (3,1) node (w1) {$w_1$};
\draw (f1) -- (w1);
\draw (f1) -- (w2);
\end{tikzpicture}	
\caption*{$m_{3}$}
\endminipage 
\caption{Three Stable Matchings\label{figure:cooperative matchings}}
\end{figure}

\paragraph{Non-Cooperative Game with Students As Actions.}
\Cref{figure:example matchings2} shows the non-cooperative game with students treated as actions. Essentially, all Nash equilibria correspond to $m_3$.

\begin{center}
\begin{figure}[h]
\hspace{0.6in} \includegraphics[scale = 1.1]{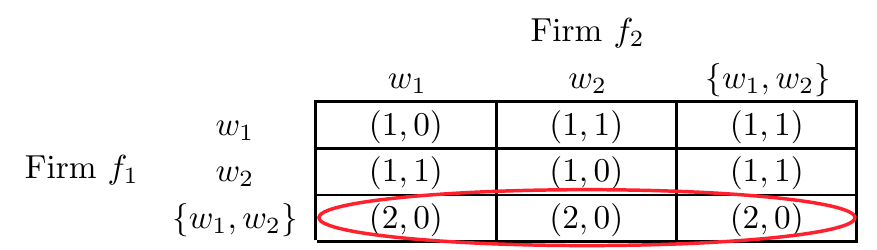}		
\caption{Students As Actions\label{figure:example matchings2}}
\end{figure}
\end{center}

Note that neither $m_1$ nor $m_2$ is a pure strategy Nash equilibrium here, because this non-cooperative game does not give agency to students in the formation of blocking coalitions.

\end{document}